\documentclass[lettersize,journal]{IEEEtran}

\usepackage[utf8]{inputenc} 
\usepackage[T1]{fontenc}
\usepackage{ifthen}
\usepackage[cmex10]{amsmath}
\usepackage{mathrsfs}
\usepackage{amsfonts}
\usepackage{amsmath}
\usepackage{amssymb}
\usepackage{amsthm}
\usepackage{multirow}
\usepackage{multicol}
\usepackage[dvipsnames]{xcolor}
\usepackage{graphicx}
\usepackage{tikz}
\usepackage{etoolbox}
\usepackage{pgfplots}
\usepackage{comment}
\usepackage[caption=false]{subfig}
\usepackage{braket}
\usepackage{stmaryrd}
\usepackage[textsize=tiny]{todonotes}
 \usepackage{url}
\pgfplotsset{compat=1.18}

\makeatletter
\renewcommand*\env@matrix[1][*\c@MaxMatrixCols c]{%
  \hskip -\arraycolsep
  \let\@ifnextchar\new@ifnextchar
  \array{#1}}
\makeatother

\usepackage[ruled, vlined, linesnumbered]{algorithm2e}

\SetCommentSty{mycommfont}
\usepackage[noadjust]{cite} 

\newtheoremstyle{bolddefinition}    
  {\topsep}                         
  {\topsep}                         
  {\itshape}                        
  {}                                
  {\bfseries}                       
  {.}                               
  { }                               
  {\thmname{#1}\thmnumber{ #2}\thmnote{ \textbf{(#3)}}} 

\theoremstyle{bolddefinition}
\newtheorem{proposition}{Proposition}

\newtheorem{theorem}{Theorem}

\theoremstyle{bolddefinition}
\newtheorem{definition}{Definition}

\newtheorem{lemma}{Lemma}

\newtheorem{corollary}{Corollary}

\newtheorem{remark}{Remark}

\usepackage[acronym,shortcuts]{glossaries}
\newacronym{AED}{AED}{automorphism-ensemble decoding}
\newacronym{CSS}{CSS}{Calderbank--Shor--Steane}
\newacronym{DPM}{DPM}{dyadic permutation matrix}
\newacronym{LER}{LER}{logical error rate}
\newacronym{LDPC}{LDPC}{low-density parity-check}
\newacronym{PCM}{PCM}{parity-check matrix}
\newacronym{QC}{QC}{quasi-cyclic}
\newacronym{QC-LDPC}{QC-LDPC}{quasi-cyclic low-density parity-check}
\newacronym{QD}{QD}{quasi-dyadic}
\newacronym{QD-LDPC}{QD-LDPC}{quasi-dyadic low-density parity-check}
\newacronym{BP}{BP$2$}{binary BP}
\newacronym{BP4}{BP$4$}{quaternary BP}
\newacronym{BP4-M}{BP$4$-M}{modified quaternary belief propagation}
\newacronym{MS}{MS}{min-sum}
\newacronym{BP+OSD}{BP$2+$OSD}{belief propagation plus ordered-statistics decoding}
\newacronym{QECC}{QECC}{quantum error correcting code}
\newacronym{QT}{QT}{quantum Tanner}
\newacronym{QEC}{QEC}{quantum error correction}
\newacronym{QLDPC}{QLDPC}{quantum low-density parity-check}
\newacronym{MSD}{MSD}{magic state distillation}
\newacronym{GB}{GB}{generalized bicycle}
\newacronym{LP}{LP}{lifted product}
\newacronym{dc}{DC}{dual-containing}
\newacronym{SO}{SO}{self-orthogonal}
\newacronym{GHP}{GHP}{generalized hypergraph product}
\newacronym{HP}{HP}{hypergraph product}
\newacronym{LHCB}{LHCB}{left-hand conveyor belt}
\newacronym{GM}{GM}{generator matrix}
\newacronym{BB}{BB}{bivariate bicycle}
\newacronym{IS}{IS}{intersecting subsets}

\pdfmapfile{+custom.map}

\newcommand{\supp}{\text{supp}}

\newcommand{\6}{\mathbf }

\usepackage{mathabx}
\usepackage{pifont}

\usepackage{array}

\newcommand{\hatstrut}{\rule{0pt}{2.9ex}} 

\begin{document}
\title{Design and Analysis of Quantum Dual-Containing CSS LDPC Codes based on Quasi-Dyadic Matrices
\thanks{
Part of the material in this paper has been presented at the 13th International Symposium on Topics in Coding (ISTC) 2025~\cite{baldelli2025quantum}. \\
This work was partially supported by the Italian National Cybersecurity Agency (ACN) under the programme for promotion of XL cycle PhD research in cybersecurity (CUP I32B24001750005).
The views expressed are those of the authors and do not represent the funding institutions.\\
The authors are with the Department of Information Engineering,
Università Politecnica delle Marche, 60131 Ancona, Italy.  
Corresponding author: Alessio Baldelli. E-mail: a.baldelli@pm.univpm.it.
}} 

\author{
\IEEEauthorblockN{Alessio Baldelli, Marco Baldi, Massimo Battaglioni, Franco Chiaraluce, and Paolo Santini \\}
}

\maketitle

\begin{abstract}
    Quantum error correcting codes are essential to achieve fault-tolerant quantum computation.
    This work introduces two constructions of high-rate, dual-containing (DC) Calderbank--Shor--Steane low-density parity-check (LDPC) codes based on quasi-dyadic matrices. 
    We identify an explicit subgroup of code automorphisms, investigate the minimum distance behavior, and provide several theoretical results on the cycle properties of these codes. 
    Monte Carlo simulations under depolarizing and phenomenological noise show better finite-length logical error rates than the considered DC benchmark codes and competitive performance against several code instances drawn from state-of-the-art quantum LDPC code families. 
    Finally, we employ an automorphism-ensemble belief propagation decoder to improve their decoding performance.
\end{abstract} 


\section{Introduction}
\label{sec:intro}

Quantum computers, due to their intrinsic properties, are highly sensitive to noise, which represents a significant limitation to their scalability.
In order to combat the effects of quantum noise and thus enable the design of scalable quantum processors, research has focused on the development of specialized \emph{\ac{QEC}} techniques, which exploit redundant physical qubits  to enable the detection and correction of errors in quantum processing systems.

In the last twenty years, many researchers have focused on a specific class of \acp{QECC}, named \emph{topological} codes~\cite{bravyi1998quantum, dennis2002topological}. 
Examples of this class are the so-called \emph{toric code}~\cite{11_kitaev1997quantum, 12_kitaev1997quantum} and its “planar” version, known as \emph{surface code}~\cite{15_freedman2001projective, 16_wang2003confinement,fowler2012surface,Roffe_Surface}.
These codes offer several favorable features, including nearest-neighbor connectivity for syndrome extraction, local error detection, and a high degree of structural symmetry. However, the implementation of these families of codes requires a large amount of \emph{physical} qubits to encode only one (surface code) or two (toric code) \emph{logical} qubit(s).

In recent years, \emph{\ac{QLDPC} codes} were proposed as a valid alternative to surface codes.
They retain many of the advantages of planar codes while supporting larger dimensions~\cite{sparse_quantum_McKay, Hagiwara2007,
Babar2015, Tillich_HGP, leverrier2022quantumtannercodes, Pant_Kal_almost_linear, Panteleev_Kala, Ostrev2024classicalproduct, IBM_BBcodes}. 
In particular, \ac{QC}-\ac{QLDPC} codes, introduced in~\cite{Hagiwara2007}, combine efficient decoding with algebraic regularity, which also simplifies the analysis of structural parameters such as \emph{girth} and \emph{minimum distance}. Cyclic and \ac{QC} codes were  used  to design state-of-the-art \ac{QLDPC} codes such as: bicycle codes \cite{sparse_quantum_McKay}, \ac{GB} codes~\cite{Hagiwara2007, kovalev_Pryad_GBcodes, Panteleev_degenerate, GB_codes_Lin}, \ac{QC} generalized hypergraph product codes~\cite{Panteleev_degenerate}, also known as \ac{LP} codes~\cite{Pant_Kal_almost_linear}, and \ac{BB} codes~\cite{IBM_BBcodes}, which are a specific instance of \ac{LP} codes.

In the classical domain, the notion of \emph{reproducible} codes stood as a generalization of the \ac{QC} framework~\cite{santini2022reproducible}. 
Their definition encompasses many code families, such as cyclic, \ac{QC}, dyadic, and \ac{QD} codes~\cite{Rajan2001}. 
In particular, similar to the cyclic case, a dyadic matrix can be reconstructed using a single row by applying \emph{dyadic permutations} (a special type of reflections, i.e., permutations of order $2$).

The idea of \emph{quantum stabilizer codes}~\cite{gottesman1997stabilizercodesquantumerror} based on dyadic matrices, in turn, was introduced in~\cite{dyadics}, where the authors discuss how this structure can be leveraged in the design of \emph{\ac{CSS} codes}~\cite{CSS, Steane_CSS}. 
However, the strict dyadic condition imposes a global structure that can limit the flexibility of code constructions. 
In particular, in~\cite{dyadics}, the authors focus on the design of quantum stabilizer codes with dimension equal to zero, which therefore do not encode any logical qubit. 
Instead, \ac{QD} codes relax this rigidity by introducing block-level dyadic symmetry~\cite{baldelli2025quantum, baldelli2026ISIT, kirsten_QD, baldelli2026QDhigh}.

\subsection{Our contribution}

We propose two constructions, \emph{Construction~A} and \emph{Construction~B}, of high-rate \ac{dc} \ac{CSS} \ac{QLDPC} codes based on \ac{QD} matrices. 
The choice of \ac{dc} codes gives identical supports for the $\6X$- and $\6Z$-type stabilizer generators. Thus, applying the physical Hadamard gate transversally preserves the code space and exchanges the two stabilizer sectors~\cite{gottesman1997stabilizercodesquantumerror}. 
This particular structure also allows for the use of the \ac{BP4-M} decoder~\cite{Babar2015}, which combines paired checks with identical supports and reduces the check-node processing relative to standard \ac{BP4}.

The \ac{QD} structure provides explicit automorphisms that we use in an automorphism-ensemble decoder.  Our finite-length simulations show competitive performance against several \ac{QLDPC} code instances with comparable parameters. 
\emph{Construction~A} also provides native redundant checks that can be used when syndrome measurements are noisy. These benefits come with structural constraints, including unavoidable length-$4$ cycles for the considered \ac{dc} code constructions. 

Recent hardware-oriented constructions such as Group-Action Lifts with Active orthogonality (GALA)~\cite{Yang2026GALA} show that high rates, girth at least $6$, and distances exceeding the stabilizer generator weight can coexist in suitable code families. 
GALA incorporates hardware compatibility into the construction itself. 
Here, we focus on the finite-length coding and decoding properties of the
\ac{dc} \ac{QD} code design approach; assessing its hardware performance would
additionally require a syndrome-extraction circuit and a physical
connectivity model, which are beyond the scope of this work.

Compared with its preliminary version~\cite{baldelli2025quantum}, the present work:
\begin{itemize}
\item provides a more flexible version of \textit{Construction~A};
\item introduces a heuristic procedure that minimizes the number of length-$4$ cycles of \textit{Construction~B} codes;
\item develops additional theoretical results that further characterize the structure, cycle properties, minimum distance behavior, and dyadic automorphisms of classical and quantum dyadic and \ac{QD} codes;
\item exploits the dyadic automorphisms of the proposed CSS codes in an AutDEC-based~\cite{AutomorphismEnsembleDecoding} performance assessment, where dyadic automorphisms are applied consistently to both qubit coordinates and stabilizer-check coordinates through the induced check action, providing a \ac{BP}+AutDEC rescue stage;
\item contains an extended numerical evaluation, including more code parameters and decoders, deeper simulations in the low \ac{LER} region, and comparisons with several code instances drawn from state-of-the-art \ac{QLDPC} code families;
\item extends the numerical analysis beyond the depolarizing-noise setting by considering a phenomenological-noise model with imperfect syndrome measurements.
\end{itemize}

\subsection{Paper outline}

The remainder of this paper is organized as follows. Section~\ref{sec:notation} introduces the notation and background, while  Section~\ref{sec:dyadic} recalls the properties of dyadic and \ac{QD} matrices. 
Section~\ref{sec:QD_LDPC} presents the two constructions, and Section~\ref{sec:cycles} studies their cycle properties. Section~\ref{sec:results} reports numerical results under depolarizing and phenomenological noise, and Section~\ref{sec:auto} those obtained with automorphism-ensemble decoding. Finally,  Section~\ref{sec:conclusions} concludes the paper.

\section{Notation and Background}
\label{sec:notation}

Upper case and lower case bold letters denote matrices, e.g., $\6A$, $\6B$, and row vectors, e.g., $\6a$, $\6b$, respectively. 
We denote transposition by $(\cdot)^T$ and inversion by  $(\cdot)^{-1}$. The $k\times k$ identity matrix is denoted as $\6I_k$, and the size is omitted when it is clear from the context. 
The same holds for the all-zero matrix $\60$ which is not necessarily a square matrix. Furthermore, let us indicate with $\mathbb{N}$ the set of natural numbers, with italic upper case letters, e.g., $A$, a generic subset of $\mathbb{N}$, with $\mathbb{C}$ the set of complex numbers, and with $\mathbb{Z}_m$ the set of all congruence classes modulo $m$. 
Moreover, we indicate with $\mathbb{F}_2$ the binary finite field, and with $\mathbb{F}_2^{n}$ an $n$-dimensional vector space over $\mathbb{F}_2$, where $n \in \mathbb{N}$. 
In addition, with calligraphic upper case letters, e.g., $\mathcal{A}$, we denote groups. We denote by $\text{rk}(\6A)$ the binary rank of the matrix $\6A$  while, using $\text{ker}(\6A)$ and $\text{row}(\6A)$, we indicate the null space of $\6A$ and its rowspace, respectively. 
The \emph{Hamming weight} (or simply \emph{weight}) of an element $\6a \in \mathbb{F}_2^n$ corresponds to the number of its non-zero entries, and it is denoted as $|\6a|$.
The \emph{support} of the vector $\6a$, which is the set of indices of its non-null entries, is denoted as $\supp(\6a)$. The \emph{order} of a vector $\boldsymbol{\alpha} \in \mathbb{F}_2^{n}$, denoted by $\#\langle \boldsymbol{\alpha} \rangle$, is the smallest positive integer $t$ such that $t \boldsymbol{\alpha} = \boldsymbol{0}$.
In the quantum domain, a \emph{quantum state} is represented using the standard bra-ket notation, i.e., $\ket{\cdot}$.

Moreover, we employ $\otimes$ for the tensor product, $\cdot$ for the dot product between vectors (i.e., $\6a\cdot\6b=\sum_i a_ib_i \bmod 2$).
In addition, we use the symbol \( \oplus \) to denote modulo-2 addition. With a slight abuse of notation, for $i, j \in \mathbb{N}$ we define $i \oplus j$ as the integer corresponding to the bit-wise modulo-2 addition of the binary representations of $i$ and $j$.

\subsection{Classical linear codes}
A binary \emph{linear code} $C[n, k]$ is a linear subspace of $\mathbb{F}_2^{n}$ with length $n$, dimension $k$, and \emph{rate} defined as $R = k/n$. 
A linear code can be represented as the kernel of a \ac{PCM} $\6H\in\mathbb{F}_2^{m \times n}$, i.e., $C = \text{ker}(\6H)$, with $m \geq n-k$. If $m = n-k = \text{rk}(\6H)$, $\6H$ is \emph{full-rank}; otherwise, it has a \emph{rank deficiency} of $m - \text{rk}(\6H)$ linearly dependent rows and we denote this \ac{PCM} as \emph{overcomplete}.
Analogously, a code can be interpreted as the rowspace of a full-rank \ac{GM} $\6G \in \mathbb{F}_2^{k \times n}$, i.e., $C = \text{row}(\6G)$. 
In other words, we have
\begin{equation*}
    C = \left\{\6m\6G\mid \6m\in\mathbb F_2^k\right\} = \left\{\6c \in\mathbb F_2^n\mid  \6H\6c^T = \60\right\},
\end{equation*}
where $\6{c} \in C$ is a codeword.
We denote with $C^{\perp} = \{ \6c \in \mathbb{F}_2^{n} \, | \, \6G \6c^T = \60 \}$ the \emph{dual code} of the code $C$. $C^{\perp}$ has length $n$ and dimension $n - k$. 
A linear code $C$ is said to be \emph{\ac{SO}} if $C \subseteq C^{\perp}$, whereas it is said to be \ac{dc} if $C^{\perp} \subseteq C$ or \emph{self-dual (SD)} if $C^{\perp} = C$.

The \emph{Hamming distance} $d_{\text{H}}(\6a, \6b)$ between two binary vectors $\6a$ and $\6b$ is the number of positions in which they differ. 
So, we call 
$d(C) = \text{min}\{ d_{\text{H}}(\6a, \6b) \, | \, \forall \, \6a \neq \6b; \, \6a, \6b \in C \}$
\emph{minimum Hamming distance} (simply, \emph{minimum distance}) of $C$. It is easy to show that $d(C)$ is equal to the  weight of the non-zero codeword(s) of minimum weight.

\subsection{Quantum Stabilizer Codes} \label{subsec:quantum_codes}

Quantum \emph{stabilizer codes}, introduced in~\cite{gottesman1997stabilizercodesquantumerror}, serve as the quantum analog of classical binary linear codes.

Let us consider the $n$-qubit Hilbert space $(\mathbb{C}^2)^{\otimes n}$, which denotes the tensor product of $n$ copies of the complex two-dimensional space $\mathbb{C}^2$.
The \emph{single-qubit Pauli operators} are 
\(
\6I,\6X,\6Y,\6Z\), 
i.e., the usual Pauli matrices~\cite{Nielsen_Chuang_2010}.  
The corresponding \emph{Pauli group} is
\(
\mathcal{P}_1 = \{\pm 1, \pm i\}\,\{\6I, \6X, \6Y, \6Z\}.
\) This group is closed under the standard matrix multiplication. Then, we define a \emph{Pauli operator on $n$ qubits} as $\6P = \sigma_0\6P_0 \otimes \, \dots \, \otimes \sigma_{n-1}\6P_{n-1} = \sigma ( \6P_0 \otimes \, \dots \, \otimes \6P_{n-1}) \in \mathbb{C}^{2^n\times 2^n}$, where $\sigma = \prod_{j = 0}^{n-1} \sigma_j \in \{ \pm1, \pm i \}$ is the \emph{global phase}, and $\6P_i \in \{ \6I, \6X, \6Z, \6Y \}$. 
The \emph{weight} $\text{wt}(\6P)$ of a Pauli operator $\6P$ is the number of non-identity components in its associated tensor product decomposition. 
So, the \emph{n-qubit Pauli group} $\mathcal{P}_{n}$ consists of all the Pauli operators that act on $n$ qubits.

Let us consider the quotient group $\mathcal{P}'_n = \mathcal{P}_n/\{ \pm \6I, \pm i\6I \}$. 
It is possible to show that there exists an isomorphism between $\mathcal{P}'_n$ and $\mathbb{F}_2^{2n}$, realized using the following mapping
\begin{equation} \label{eq:mapping}
   \6P = \bigotimes_{i = 0}^{n-1} \6X^{x_i}_i \6Z^{z_i}_i \mapsto [\6x \, | \, \6z]  = [x_0, \dots, x_{n-1} \, | \, z_0, \dots, z_{n-1}],
\end{equation}
$\forall \, \6P \in \mathcal{P}'_n$, where the index $i$ represents the $i$-th qubit and $x_i, z_i \in \mathbb{F}_2$. In other words, \eqref{eq:mapping} establishes a one-to-one correspondence between Pauli operators on $n$ qubits and vectors in $\mathbb{F}_2^{2n}$.
Consider $\6P, \tilde{\6P} \in \mathcal{P}'_n$ and their associated binary representations $[\6x \, | \, \6z], [\tilde{\6x} \, | \, \tilde{\6z}]$ using the above mapping. If we find $\6x \cdot \tilde{\6z} \, \, +  \, \, \6z \cdot \tilde{\6x} = \60$, the two Pauli operators commute.

Let us introduce $\tilde{\mathcal{P}}_n$ as the subset of the Pauli operators $\6P \in \mathcal{P}_n$ with $\sigma = 1$.
A \emph{stabilizer group} $\mathcal{S}$ is a commutative subgroup of $\mathcal{P}_n$ such that $-\6I \notin \mathcal{S}$.
Suppose $\mathcal{S}$ has $n-k$ \emph{independent} generators
$\6S_0, \dots, \6S_{n-k-1} \in \tilde{\mathcal{P}}_n$, called the \emph{stabilizer generators},
so that $\mathcal{S} = \langle \6S_0,\dots, \6S_{n-k-1} \rangle$.
By independence of the generators we mean that, up to the global phase factor $\sigma$, it is not possible to obtain one of them from the others using the group operation, i.e., matrix multiplication.

A \emph{quantum stabilizer code} $\mathcal{C} \llbracket n, k \rrbracket$ is defined as a $2^{k}$-dimensional subspace of $(\mathbb{C}^{2})^{\otimes n}$ such that every quantum state $\ket{\psi} \in \mathcal{C}$ is stabilized by all the stabilizer generators, that is,
\begin{equation*}
    \mathcal{C} = \{ \ket{\psi} \in (\mathbb{C}^{2})^{\otimes n} \, \, | \, \, \6S_i \ket{\psi} = \ket{\psi}, i \in \{0, \dots, n-k-1 \} \}.
\end{equation*}
Then, $\mathcal{S}$ is called \emph{stabilizer group} of $\mathcal{C}$. As an example, for $i = 1$, we obtain $\6S_1 \ket{\psi} = \ket{\psi}$, and we say that $\ket{\psi}$ is \emph{stabilized} by $\6S_1$. 
The minimum distance of $\mathcal{C}$ is defined as follows
\begin{align*}
     d(\mathcal{C}) = \text{min} \{ & \text{wt}(\6P) \, | \, \6x \cdot \6z_{j} + \6z \cdot \6x_{j} = 0; \\ 
    & \6P \in \tilde{\mathcal{P}}_n, \, \6P \notin \langle i \, \6I \rangle \, \mathcal{S}, \, j \in \{0, \dots, n-k-1 \} \},
\end{align*}
where $[\6x \mid \6z]$ and $[\6x_j \mid \6z_j]$ denote the binary vectors associated with
$\mathbf{P}$ and with the stabilizer generator $\mathbf{S}_j$, respectively, and $\langle i \, \6I \rangle \, \mathcal{S}$ denotes the set of all global-phase multiples of stabilizers. 

Any Pauli operator $\6P \in \tilde{\mathcal{P}}_n$ can be used to represent a \emph{Pauli error} on $n$ qubits. It acts on the quantum system and, under its effect, such a system goes from a state $\ket{\psi}$ to a state $\6P \ket{\psi}$. However, for every quantum stabilizer code $\mathcal{C}$, it holds that $\6P \ket{\psi} = \ket{\psi}$ for all $\ket{\psi} \in \mathcal{C}$ if and only if $\6P \in \mathcal{S}$. 
Thus, there exist Pauli errors $\6P \in \tilde{\mathcal{P}}_n$  which do not corrupt $\ket{\psi}$, referred to as \emph{degenerate} Pauli errors~\cite{Panteleev_degenerate}.  Instead, a Pauli error that commutes with all the stabilizer generators but does not belong to $\langle i \, \6I \rangle \, \mathcal{S}$ cannot be detected and alters the encoded information; it is referred to as a \emph{logical error}.

If we apply the mapping \eqref{eq:mapping} to all the $n-k$ stabilizer generators of a code $\mathcal{C} \llbracket n, k \rrbracket$, we obtain the following matrix
\begin{equation} \label{eq:stab_matrix}
    \6H_{\text{Q}} = 
    \begin{bmatrix}[c|c]
        \6H_{\text{X}} & \6H_{\text{Z}}
    \end{bmatrix}
    \in \mathbb{F}_2^{(n-k) \times 2n},
\end{equation}
where $\6H_{\text{X}}, \6H_{\text{Z}} \in \mathbb{F}_2^{(n - k) \times n}$. 
Eq.~\eqref{eq:stab_matrix} represents the \ac{PCM} of the stabilizer code $\mathcal{C}$, where its rows are the binary representation of its stabilizer generators.  It is possible to convert the commutative constraint of the stabilizer generators 
into the orthogonality condition of the rows of $\6H_{\text{Q}}$, expressed by the following \emph{symplectic product}
\begin{equation} 
\label{eq:symp}
    \6H_{\text{X}} \6H_{\text{Z}}^{T} + \6H_{\text{Z}} \6H_{\text{X}}^{T} = \60.
\end{equation}

In the present work, we are interested in studying \ac{CSS} codes, an important subclass of stabilizer codes~\cite{CSS, Steane_CSS}. 
They are characterized by stabilizer generators that are either tensor products of only $\6I$ and $\6X$ ($\6X$-type stabilizers) or of only $\6I$ and $\6Z$ operators ($\6Z$-type stabilizers). 
A \ac{CSS} code $\mathcal{C} \llbracket n, k_1 - k_2 \rrbracket$ is constructed based on two classical binary  codes $C_1 [n, k_1]$ and $C_2 [n, k_2]$, with $k_1 \geq k_2$, and such that $C_2 \subseteq C_1$.
It follows that the symplectic product relation (\ref{eq:symp}) in this case becomes $\6{H}_1\6{G}_2^{T} = \60$ (or $\6{G}_2\6{H}_1^{T} = \60$). This means that, in such codes, any row of $\6H_1$ is orthogonal to all  rows of $\6G_2$ (and vice versa).  
In particular, the code $C_1$ is used to correct $\6Z$ errors, i.e., phase-flip errors; instead, the code $C_2^{\perp}$ is used for $\6X$ errors, i.e., bit-flip errors.
The  minimum distance of the \ac{CSS} code $\mathcal{C}$ is 
\begin{equation}
\label{eq:quanutm_dmin_CSS}
    d(\mathcal{C}) = \text{min} \{ d_{\text{Z}}, d_{\text{X}} \},
\end{equation}
where
\begin{equation*}
    d_{\text{Z}} = \min_{ \6z \in C_1 \setminus C_2} |\6z| \geq d(C_1), \, \, \, \, \, \,
    d_{\text{X}} = \min_{\6x \in C_2^{\perp} \setminus C_1^{\perp}} |\6x| \geq d(C_2^{\perp}),
\end{equation*}
then, $d(\mathcal{C}) \geq \text{min} \{ d(C_1), d(C_2^\perp) \}$.
Moreover, its rate is $R_{\text{Q}} = (k_1 - k_2)/n$.
For \ac{CSS} codes, the correction of $\6X$ and $\6Z$ errors can be performed separately.

If the \ac{CSS} code is constructed from a single classical $[n,k]$ linear code $C$ which satisfies $C^{\perp} \subseteq C$, i.e., a \ac{dc} code, then we obtain a  subclass known as \ac{dc} \ac{CSS} codes. 
In this setting, the two classical codes defining the \ac{CSS} construction are $C$ and its dual $C^{\perp}$, which share the same length $n$ and have dimensions $k_1 = k$ and $k_2 = n-k$, respectively. 
Therefore, the \ac{PCM} $\mathbf{H}$ associated to $C$ satisfies the following orthogonality condition 
\begin{equation}
    \6{H}_1\6{G}_2^{T} = \6{H}\6{H}^{T} = \6{0}.
\label{eq:orto}
\end{equation}
The resulting \ac{CSS} code $\mathcal{C}$ has length $n$ and dimension $2k - n$. 
Thus, for a \ac{dc} \ac{CSS} code, the quantum minimum distance is
\[
d(\mathcal C) = \min_{\mathbf c \in C\setminus C^\perp} |\6c|\geq d(C).
\] 
In particular, the rows of $\mathbf H$ belong to $C^\perp$ and therefore do not provide an upper bound on $d(\mathcal C)$.

\subsection{Classical and Quantum LDPC Codes}
In the \emph{classical} domain, a linear code is called a \ac{LDPC} code~\cite{gallager1962low} if its \ac{PCM} $\6H \in \mathbb{F}_2^{m \times n}$ is \emph{sparse}.  
Every \ac{PCM} can be associated with a Tanner graph~\cite{Tanner1981}, which is a bipartite graph consisting of $n$ \emph{variable} nodes and $m$ \emph{check} nodes. 
A simple cycle in a Tanner graph is a closed path with distinct edges and no repeated vertices apart from the initial and final vertex. Its length is even and at least $4$, because the graph is bipartite and has no multiple edges. We only consider simple cycles and, in the following, we call them cycles for brevity.
The \emph{girth} $g$ is defined as the length of the shortest cycle(s) in the graph. 
Throughout the paper, we denote with $n_4$ the number of length-$4$ cycles in the Tanner graph of a certain code.
In the following, we use the terms girth of a Tanner graph, of a code, and of a \ac{PCM} interchangeably.
If all the variable (check) nodes have the same degree, the code is said to be regular in the \emph{column (row) weight} and we adopt $\lambda_{\text{c}}$ ($\lambda_{\text{r}}$) to denote the column (row) weight of its \ac{PCM}. If the code is not regular in the column (row) weight, then we use $\lambda_{\text{c, avg}}$ ($\lambda_{\text{r, avg}}$) to denote the average column (row) weight.

In the \emph{quantum} domain, a  stabilizer code $\mathcal{C}$ represented by a sparse \ac{PCM} is called \ac{QLDPC} code.  
Analogously, we can refer to them as \emph{sparse stabilizer codes}. As above, it is possible to describe a stabilizer code $\mathcal{C}$ with a Tanner graph. 
In this case, the variable nodes represent the $n$ qubits, whereas the check nodes are associated to the stabilizer generators. 
In the case of \ac{CSS} codes, the Tanner graph associated to $\6H_{\text{Q}}$ consists of two independent Tanner graphs corresponding to the matrices $\6H_1$ and $\6G_2$.
The row weight of $\6H_1$ ($\6G_2$) represents the weight of $\6X$-type ($\6Z$-type) stabilizer generators.

\subsection{Role of code automorphisms in \ac{CSS} codes} \label{subsec:fault-auto}
  
The automorphisms of a classical code are the permutations that maintain the code space. More specifically, the automorphism group $\mathcal{A}(C)$ on the classical code $C$ is the set of all permutations $\pi \in S_n$, where $S_n$ is the symmetric group of $n$ elements, that preserves the code space, i.e., such that $\pi(C) = \left\{\pi^{-1}(\6c)\mid \6c\in C\right\} = C$. In other words, it holds that: $\pi^{-1}(\6c) = (c_{\pi(0)}, c_{\pi(1)}, \dots, c_{\pi(n-1)}) \in C, \, \forall \, \6c = (c_0, c_1, \dots, c_{n-1}) \in C$.
It is known from~\cite{MacWilliams_Sloane} that $\mathcal{A}(C) = \mathcal{A}(C^{\perp})$.
For a \ac{CSS} code $\mathcal{C}$ based on one 
\ac{SO} code, the relevant automorphism group for $\mathcal{C}$ is $\mathcal{A}(C)$~\cite{Grassl_FT}.

For a \ac{dc} \ac{CSS} code $\mathcal{C}$ constructed from a \ac{dc} code $C$, 
the relevant automorphism group is 
$\mathcal{A}(\mathcal{C}) = \mathcal{A}(C)$, similarly to the case of one \ac{SO} code. More generally, when a \ac{CSS} code is constructed from two distinct codes $C_1$ and $C_2$, the relevant automorphism group is $\mathcal{A}(\mathcal{C}) = \mathcal{A}(C_1)\cap\mathcal{A}(C_2)$. 

For a \ac{dc} \ac{CSS} code, every coordinate automorphism of the underlying classical code preserves both \ac{CSS} components. Thus, any explicitly known subgroup of classical automorphisms can be used directly, without determining an intersection between two groups. In particular, the dyadic permutations identified in Section~\ref{sec:dyadic} provide such a subgroup.

In the present work, we will use this automorphism structure primarily to characterize the algebraic symmetries of the proposed DC CSS constructions and, later, to define the coordinate permutations used in the automorphism-ensemble decoding assessment of Section~\ref{sec:auto}.

\subsection{Noise models} \label{app:noise}
In this paper, we consider two noise models: \emph{depolarizing noise} and \emph{phenomenological noise}.

In the depolarizing-noise model, which relies on perfect syndrome measurements, any qubit is affected by a Pauli error $\mathbf{X}$, $\mathbf{Z}$, or $\mathbf{Y}$ with probability $p/3$ each, and remains error-free with probability $1 - p$. The parameter $p$ is called the \emph{depolarizing probability}. 
The state $\ket{\psi} \in (\mathbb{C}^{2})^{\otimes n}$ is corrupted by a Pauli error $\6E = \bigotimes_{i = 0}^{n-1} \6E_i \in \tilde{\mathcal{P}}_n$, where $\6E_i\in \{\6I, \6X, \6Z, \6Y \}$.

Since, in actual quantum systems, syndrome measurement operations are inherently noisy, fault-tolerant \acp{QECC} need to detect and correct errors even when they corrupt the syndrome readout.
While circuit-level noise simulations provide the most accurate noise setting, a useful simplified model is the second model considered here, i.e., phenomenological noise~\cite{Gottesman_FT, Gottesman_introduction}. 
Unlike the depolarizing-noise model, in this setting we also consider syndrome readout errors as independent random events with readout error probability $\epsilon$. 
In particular, we adopt the single-shot syndrome measurement scenario, which is the framework used in~\cite[Section~$5$]{Ostrev2024classicalproduct}. 
Therefore, we leverage the linear dependencies among the rows of the \ac{PCM}, called \emph{meta-checks}, to mitigate the negative effect of syndrome readout errors.
In other words, the binary decoder of a \ac{CSS} code is fed with the following \emph{extended \ac{PCM}}
\begin{equation*}
    \widetilde{\6H} =
    \begin{bmatrix} 
        \6H & \6I \\ 
        \60 & \6L
    \end{bmatrix},
\end{equation*}
where $\6H \in \mathbb{F}_2^{m \times n}$ is the \ac{PCM} representing either $C_1$ or $C_2^{\perp}$ and $\6L \in \mathbb{F}_2^{(m - r) \times m}$; then, we have $\6I \in \mathbb{F}_2^{m \times m}$, $\60 \in \mathbb{F}_2^{(m - r) \times n}$, where $r = \text{rk}(\6H)$.
Specifically, $\6L$ is the full-rank \emph{meta-check matrix} representing the classical code $C_{\text{L}} = \text{ker}(\6L)$. It holds that 
\begin{equation} \label{eq:meta_matrix}
    \text{ker}(\6L) = \text{row}(\6H^T),
\end{equation}
so that $\6L \6H = \60$.
In particular, the minimum distance of $C_{\text{L}}$, that with a slight abuse of notation we denote as $d_{\text{L}}$, is called \emph{meta-check distance}. Every nonzero syndrome-error pattern of weight smaller than $d_{\text{L}}$ violates at least one meta-check and is therefore detectable. By~\eqref{eq:meta_matrix}, $d_{\text{L}}$ is upper bounded by the minimum column weight of $\6H$, since each column of $\6H$, i.e., the syndrome of a single-qubit error, is a codeword of $C_{\text{L}}$; we denote such a bound as $U(d_{\text{L}})$. 
If $\mathbf{H}$ has full row rank, $\mathbf{L}$ has no rows and $C_\mathrm{L}=\mathbb{F}_2^m$. Hence $d_{\text{L}}=1$, and no meta-check redundancy is available.

We anticipate that, in Section~\ref{sec:results}, we use the native \ac{PCM} produced or reported for each code instance, without adding or removing linearly dependent rows. Consequently, a code whose native \ac{PCM} is full rank has no meta-check redundancy in the considered experiment.

\section{Basics on Dyadic and Quasi-Dyadic Codes}
\label{sec:dyadic}

A matrix $\6M$ over $\mathbb{F}_2$ is defined as \emph{reproducible} if it can be completely described by a small subset of its rows (called the \emph{signature set}) and a set of linear transformations that act on that subset of rows \cite{santini2022reproducible}. 
Moreover, $\6M$ is \emph{quasi-reproducible} if it is an array of reproducible matrices. A classical linear code $C \subseteq \mathbb{F}_2^{n}$ is reproducible (quasi-reproducible) if it can be represented by a reproducible (quasi-reproducible) \ac{GM} or \ac{PCM}. 
A special class of reproducible matrices is that of the so-called  \emph{dyadic} matrices. For these matrices, the signature $\6m$ is given by their first row.
\begin{definition}[Ring of dyadics] \label{def:dyadic}
    Given an integer $\ell\geq 1$, we define $\mathcal M_{\ell}(\mathbb F_2)$ as the set of $2^\ell\times 2^\ell$ matrices with entries over $\mathbb F_2$ and structured as follows
    \begin{equation}
    \label{eq:dyadics}
        \6M = \begin{bmatrix}
        \6A & \6B \\
        \6B & \6A
        \end{bmatrix}, \quad \6A,\6B\in\mathcal M_{\ell-1}(\mathbb F_2).
    \end{equation}
    For $\ell = 0$, $\mathcal M_0(\mathbb F_2)\coloneq \mathbb F_2$.
    For any $\ell$, when equipped with standard matrix sum and multiplication,  $\mathcal M_\ell(\mathbb F_2)$ forms a commutative ring.
\end{definition}

It is easy to verify that dyadic matrices are symmetric, and that, coherent with its definition, any dyadic matrix can be fully described by using only its signature $\6m = (m_{0,0}, m_{0,1}, \cdots, m_{0, 2^{\ell}-1})$.
In fact, for any $i, j$, we have $m_{i, j} = m_{0, i\oplus j}$. By weight of a dyadic matrix, we refer to the weight of its signature, that is  $|\6m|$.  
Let us recall some useful theoretical results on dyadic matrices.

\begin{proposition}
    If $\6M \in \mathcal M_{\ell}(\mathbb F_2)$ is dyadic and its signature $\6m$ has odd weight, then:
    \begin{itemize}
        \item[$(i)$] $\6M$ is not singular;
        \item[$(ii)$] $\6M$ is invertible and the inverse matrix equals $\6M$;
        \item[$(iii)$] $\6M\6M^{T} = \6I$.
    \end{itemize}
\end{proposition}
The proof is available in \cite{banegas2020designing}.

Let $\mathcal D_\ell(\mathbb F_2) = \{ \6D^{(0)}\coloneq \6I_{2^\ell}, \6D^{(1)},\cdots, \6D^{(2^\ell-1)} \}\subseteq \mathcal M_\ell(\mathbb F_2)$ be the set containing all dyadic matrices with weight 1, so that $\6D^{(i)}$ is the dyadic matrix whose signature has a $1$ in position $i \in \{0,\cdots, 2^\ell-1\}$. 
Each $\6D^{(i)}$ is a \emph{\ac{DPM}}. We define the  bijection 
\[
\phi : \{0,\dots,2^\ell-1\} \rightarrow \mathbb F_2^\ell
\]
that maps each index $i$ to its 
$\ell$-bit binary representation \[
\phi(i) = (b_{\ell-1},\dots,b_1,b_0), 
\qquad 
i = \sum_{k=0}^{\ell-1} b_k\,2^k,\; b_k \in \mathbb F_2.
\]  Note that each \ac{DPM} (excluding $\6I$) has multiplicative order $2$, since
$${(\6D^{(i)})}^2 = \6D^{(i)} 
(\6D^{(i)})^T = \6D^{(i)} 
(\6D^{(i)})^{-1} = \6I_{2^\ell},\quad \forall \, i.$$ 
It is easy to see that any dyadic matrix can be uniquely expressed as a linear combination of \acp{DPM}, that is,
$$\exists! \, \, \6m \in \mathbb F_2^{2^\ell}: \6M = \sum_{i = 0}^{2^\ell-1} m_i\6D^{(i)},\,\,\, \forall \, \6M \in \mathcal M_\ell(\mathbb F_2).$$
Furthermore, it is straightforward to note that, for any two \acp{DPM} $\6D^{(i)}$ and $\6D^{(j)}$, it holds that $\6D^{(i)}
\6D^{(j)} = \, \6D^{(i\oplus j)}$.

The next lemma naturally follows.

\begin{lemma}\label{lem:rank}
Let $\ell \geq 1$ be an integer and $\6M\in\mathcal M_\ell(\mathbb F_2)$ with signature $\6m$. Then
$$\6M^2 = 
\begin{cases}
    \60, \, \, \, \, \, \, \, \, \text{if }|\6m|\text{ is even,}
   \\\6I_{2^\ell}, \, \, \, \, \text{if }|\6m|\text{ is odd.}
\end{cases}$$
Hence, $\6M$ is non singular if and only if its signature has odd weight and, in such a case, $\6M^{-1} = \, \6M$.
\end{lemma}

A matrix is said to be \emph{quasi-dyadic (\ac{QD})} if it is composed only of dyadic blocks, and a linear code is referred to as \ac{QD} if it is possible to represent it using a \ac{QD} \ac{PCM} or \ac{GM}.

\begin{proposition}
    For any dyadic code $C$, represented by a dyadic \ac{PCM}, 
    its automorphism group contains all the $2^\ell$ \acp{DPM} $\6D^{(i)}$, $ \forall i \in \{ 0, \dots, 2^{\ell} - 1 \}$.
    Then, for any dyadic code, the automorphism group has size at least $2^\ell$, i.e., $|\mathcal{A}(C)| \geq 2^{\ell}$.
\end{proposition} 

Next we present a similar result for \ac{QD} codes.

\begin{proposition}
    Let $C=\ker(\6H)$, where $\6H\in\mathbb F_2^{m\times n}$ is a \ac{QD} \ac{PCM} whose blocks are dyadic matrices of size $2^\ell$.
    Then, the following permutations are automorphisms of $C$:
    \begin{equation*}
        \6P_i = \6I_{n/2^{\ell}} \otimes \6D^{(i)}, \, \, \forall i \in \{ 0, \dots, 2^{\ell} - 1 \};
    \end{equation*} 
    so, $|\mathcal{A}(C)| \geq 2^{\ell}$.
\label{propo:automtype}
\end{proposition}

\begin{IEEEproof}
Right-multiplying $\6H$ by $\6P_i$ multiplies each dyadic block $\6M$ of $\6H$ by $\6D^{(i)}$, and $\6M\6D^{(i)}=\6D^{(i)}\6M$ since dyadic matrices commute. Hence,
\begin{equation*}
    \6H\6P_i = \big(\6I_{m/2^\ell} \otimes \6D^{(i)}\big)\6H.
\end{equation*}
Since $\6P_i$ is symmetric, for every $\6c\in C$ we obtain $\6H(\6c\6P_i)^T=\6H\6P_i\6c^T=(\6I_{m/2^\ell}\otimes\6D^{(i)})\6H\6c^T=\60$, i.e., $\6P_i$ maps $C$ onto itself. The $2^\ell$ permutations $\6P_i$ are distinct, since the \acp{DPM} $\6D^{(i)}$ have such property.
\end{IEEEproof}

\section{Dual-Containing \ac{QD}-\ac{QLDPC} Codes}
\label{sec:QD_LDPC}

In this section, we present two new constructions of the \ac{PCM} of classical \ac{QD} codes as $w\times u$ arrays of dyadic matrices, and show how to use them to design \ac{dc} \ac{CSS} codes.

The \emph{orthogonality condition} is expressed by~\eqref{eq:orto}, which characterizes \ac{dc} \ac{CSS} codes, while the \emph{cycles condition} regards the avoidance of length-$4$ cycles. 
In particular, we know that there is a length-$4$ cycle in $\6H$ if there exists at least one pair of rows or columns that share at least two non-zero positions.

\begin{theorem} \label{theo:girt_DC}
    Let us consider an  $[n, k]$ classical code $C \subseteq \mathbb{F}_2^{n}$, represented by $\6H \in \mathbb{F}_2^{m \times n}$ with column weight $v > 1$. 
    Then, it cannot satisfy both the orthogonality and the cycles conditions.
\end{theorem}
\begin{IEEEproof}
To satisfy  the orthogonality condition, any possible pair of rows in $\6H$ should share a null or even number of positions occupied by a “$1$”. Instead, the absence of length-$4$ cycles in the Tanner graph requires that any two distinct rows of $\6H$ share \emph{at most
one} common non-zero position. Combining the two constraints, the number of common non-zero positions must be both even and at most one, hence it must be $0$. 
Therefore, the supports of all the rows of $\6H$ must be pair-wise disjoint, which implies that every column has weight at most $1$, i.e., $v = 1$. 
This contradicts the hypothesis of the theorem, which assumes $v > 1$.
Hence, $C$ cannot satisfy both the orthogonality and the cycles conditions simultaneously.
\end{IEEEproof}

\begin{theorem}
\label{theo:distance}
    Let $\mathbf H\in\mathbb F_2^{m\times n}$ be the \ac{PCM} of a classical \ac{dc} code $C$, regular in the row weight $\lambda_{\emph{r}}$.
    Then $d(C)\leq\lambda_{\emph{r}}$.
\end{theorem}

\begin{IEEEproof}
    The PCM $\6H$ of $C$ is the \ac{GM} of $C^{\perp}$. 
    This means that the rows of $\6H$ with weight $\lambda_{\text{r}}$ are codewords of $C^{\perp}$.
    Since it holds that $C^{\perp} \subseteq C$, the codewords of $C^{\perp}$ are also codewords of $C$.
    Then, $d(C)$ is upper-bounded by $\lambda_{\text{r}}$.
\end{IEEEproof}

We now present some properties of \ac{QD} \acp{PCM} composed exclusively of \acp{DPM}.

\begin{proposition}
    \label{prop:rank_deficiency}
    Let $\6H \in\mathbb{F}_2^{2^\ell w\times 2^\ell u}$ be a $w \times u$ array of \acp{DPM}. Then
    \[
    \mathrm{rk}(\6H) \le w2^{\ell} - (w-1).
    \]
\end{proposition}
\begin{IEEEproof}
    Each \ac{DPM} (as any permutation matrix) has the property that the sum (over $\mathbb{F}_2$) of all its rows equals the all-one vector of length $2^{\ell}$. Consequently, for any fixed block-row of $\6H$, the sum of the $2^{\ell}$ corresponding rows is the all-one row. Since $\6H$ has $w$ block-rows, the sums of the rows associated with different block-rows are identical. Therefore, among these $w$ row sums, at least $w-1$ are linearly dependent.
\end{IEEEproof}

\begin{remark} \label{rem:expected_parameters}
    Consider  a $w \times u$ array of \acp{DPM} $\6H \in\mathbb{F}_2^{2^\ell w\times 2^\ell u}$, representing the code $C$. 
    According to Proposition~\ref{prop:rank_deficiency}, the associated classical code $C$ has dimension \emph{at least} $k = (u - w)2^{\ell}$, and therefore the rate satisfies $R \ge 1 - w/u$.
    Since $\6H$ is a \ac{QD} matrix formed by \acp{DPM}, the code is regular in the column weight, $\lambda_{\emph{c}} = w$, and in the row weight, $\lambda_{\emph{r}} = u$.  
    Moreover, the dimension of the associated \ac{dc} \ac{CSS} code is, at least, $2k - n 
    = (u - 2w)2^{\ell}$, and the weight of its stabilizer generators is $u$. 
    Then, if we aim to have a \ac{CSS} code with \emph{expected} dimension greater than $0$, we need $w < u/2$.
\label{rem:parconstrA}
\end{remark}

The following original \ac{LDPC} code constructions will be used to design component codes of \ac{dc} \ac{CSS} codes.

\subsection{Construction A}
\label{subsec:A}

\emph{Construction~A} returns a \ac{PCM} $\mathbf H \in\mathbb{F}_2^{2^\ell w\times 2^\ell u}$ which is a $w\times u$ array of \acp{DPM}.
We assume
\(
u=2^m
\)
for some integer $m> 1$, and we identify the block-column indices
 as
\(
j\in\{0,1,\ldots,u-1\}.
\)
The construction proceeds as follows.
\begin{enumerate}
    \item Select $u\leq 2^\ell$ distinct \acp{DPM}
    \[
    \mathbf Q_0,\mathbf Q_1,\ldots,\mathbf Q_{u-1}
    \in \mathcal D_\ell(\mathbb F_2).
    \]
    These matrices form the reference block-row
    \[
    \mathbf h_0=
    \big[
    \mathbf Q_0\ \mathbf Q_1\ \cdots\ \mathbf Q_{u-1}
    \big].
    \]
    \item Select $w\leq u$ distinct shift values
    \[
    a_0,a_1,\ldots,a_{w-1}
    \in\{0,1,\ldots,u-1\},
    \]
    with $a_0=0$.
    \item For each $i\in\{0,1,\ldots,w-1\}$, define the $i$-th
    block-row
    \[
    \mathbf h_i=
    \big[
    \mathbf D_{i,0}\ \mathbf D_{i,1}\ \cdots\ \mathbf D_{i,u-1}
    \big]
    \] by setting
    \[
    \mathbf D_{i,j}=\mathbf Q_{j\oplus a_i},
    \qquad j=0,1,\ldots,u-1.
    \]
\end{enumerate}
The \ac{PCM} $\mathbf H$ is obtained by vertically stacking the block-rows $\mathbf h_0,\mathbf h_1,\ldots,\mathbf h_{w-1}$.

\begin{proposition}
Let $\mathbf H$ be obtained using \emph{Construction~A}. Then
\[
\mathbf H\mathbf H^T=\mathbf 0.
\]
\end{proposition}

\begin{IEEEproof}
We prove that any two block-rows of $\mathbf H$ are orthogonal. First, consider the product of a block-row with itself. Since each
$\mathbf D_{i,j}$ is a \ac{DPM}, we have
\(
\mathbf D_{i,j}\mathbf D_{i,j}^T=\mathbf I.
\) 
Therefore,
\[
\mathbf h_i\mathbf h_i^T
=
\sum_{j=0}^{u-1}
\mathbf D_{i,j}\mathbf D_{i,j}^T
=
u\mathbf I
=
\mathbf 0,
\]
because $u=2^m$ is even and the operations are over $\mathbb F_2$. Now consider two distinct block-rows $\mathbf h_i$ and $\mathbf h_s$,
with $i\neq s$. By construction,
\[
\mathbf h_i\mathbf h_s^T
=
\sum_{j=0}^{u-1}
\mathbf Q_{j\oplus a_i}
\mathbf Q_{j\oplus a_s}^{T}.
\]
Let \(
c=a_i\oplus a_s\). Since $a_i\neq a_s$, we have $c\neq 0$. The map
$j\mapsto j\oplus c$ has no fixed points and partitions the set
$\{0,1,\ldots,u-1\}$ into disjoint pairs
$\{j,j\oplus c\}$. 
For each such pair, the corresponding two terms in
$\mathbf h_i\mathbf h_s^T$ are
\[
\mathbf Q_{j\oplus a_i}
\mathbf Q_{j\oplus a_s}^{T}
\oplus
\mathbf Q_{(j\oplus c)\oplus a_i}
\mathbf Q_{(j\oplus c)\oplus a_s}^{T}.
\]
Using $c=a_i\oplus a_s$, these two terms have the form
\[
\mathbf Q_A\mathbf Q_B^T + \mathbf Q_B\mathbf Q_A^T
\]
for suitable \acp{DPM} $\mathbf Q_A$ and $\mathbf Q_B$. Since
\acp{DPM} are symmetric and commute, this becomes
\[
\mathbf Q_A\mathbf Q_B + \mathbf Q_B\mathbf Q_A
=
2\mathbf Q_A\mathbf Q_B
=
\mathbf 0
\]
Hence all terms cancel pairwise, and \(
\mathbf h_i\mathbf h_s^T=\mathbf 0\). Therefore all block-rows are mutually orthogonal, and
\(
\mathbf H\mathbf H^T=\mathbf 0.
\)
\end{IEEEproof}

We  remark that the final \ac{PCM} does not need to contain all the $u$ block-rows associated with the possible XOR shifts. 
Rather, any subset of $w\leq u$ distinct shift-generated block-rows may be selected. 
All classical and \ac{dc} \ac{CSS} codes obtained in this way are referred to as \emph{Construction~A} codes. 
We now provide an upper bound on the minimum distance of the underlying classical codes.

\begin{corollary}
    Let $C$ be a \emph{Construction~A} code for some $u$. 
    Then, it holds that $d(C) \leq u$.
    \label{cor:distcosta}
\end{corollary}
\begin{IEEEproof}
    Immediately follows from Theorem~\ref{theo:distance}.
\end{IEEEproof}

Corollary~\ref{cor:distcosta} bounds the classical minimum distance $d(C)$. For fixed $u$, this bound is independent of the lifting size, but it does not imply an upper bound on the quantum minimum distance $d(\mathcal C)$.

Table~\ref{tab:constrA_distance} reports finite-length examples with fixed dyadic-block size and comparable quantum rates. We denote by $\widehat d(\mathcal C)$ the quantum-distance estimate obtained by searching for nontrivial logical operators with the tool in~\cite{Pryadko_2022}. Each reported value is an upper bound on $d(\mathcal C)$. For $\mathcal C_{\mathrm A}\llbracket512,96\rrbracket$ and $\mathcal C_{\mathrm A}\llbracket1024,172\rrbracket$, the estimates are $8$ and $16$, respectively.

\begin{table}[t]
\centering
\caption{Finite-length examples of Construction~A codes with fixed dyadic-block size $2^\ell$ and comparable coding rates.}
\label{tab:constrA_distance}
\begin{tabular}{|c|c|c|c|c|c|}
\hline
\textbf{Code} & \multicolumn{5}{c|}{\textbf{Parameters}} \\
\cline{2-6}
 & $\ell$ & $u$ & $w$ & $R_{\text{Q}}$ & $\widehat d(\mathcal C) \hatstrut$ \\
\hline
$\mathcal{C}_{\text{A}}\llbracket512,96\rrbracket$    & $6$ & $8$  & $5$  & $0.19$ & $8$  \\
\hline
$\mathcal{C}_{\text{A}}\llbracket1024,172\rrbracket$  & $6$ & $16$ & $8$  & $0.17$ & $16$ \\
\hline
\end{tabular}
\end{table}

\subsection{Construction~B} 

\emph{Construction~B} arises from the following theorem, which refines a result in~\cite{baldelli2025quantum}. 
\begin{theorem} \label{theo:oneRow}
    Let $u\geq 2$ be an integer and $\6H \in \mathbb F_2^{2^\ell\times u 2^\ell}$ be a \ac{QD} matrix with the following structure 
    \begin{equation} \label{eq:con_B}
    \6H = 
    \begin{bmatrix}
        \6M_0 & \6M_1 & \dots & \6M_{u-1}
    \end{bmatrix}, \, \,  \6M_i\in\mathcal M_\ell(\mathbb F_2),
    \end{equation}
     $\forall \, i \in \{ 0, \dots, u-1 \}$, where each $\6M_i$ is a dyadic matrix with side $2^\ell$ and odd weight $v$. 
    Let $C\subseteq \mathbb F_2^{u2^\ell}$ be the classical code whose \ac{PCM} is $\6H$. Under such assumptions:
\begin{itemize}
    \item[$(i)$] $\6H$ is full-rank; hence, $n - k = 2^\ell$ and $k = 2^\ell(u - 1)$;
     \item[$(ii)$] if no two blocks $\6M_i$ and $\6M_j$, with $i\neq j$, have a row in common, then $C$ contains at least $\binom{u}{2}2^{\ell+1}$ codewords of weight $2v$;
    \item[$(iii)$] if $u$ is even, then $C^{\perp} \subseteq C$, i.e., the code is \ac{dc}.
\end{itemize}
\end{theorem}
\begin{IEEEproof}
By Lemma~\ref{lem:rank}, each block $\6M_i$ is invertible, since $v$ is odd; hence, the rows of $\6H$ are linearly independent, which proves $(i)$.

For $(ii)$, fix $0\leq i<j\leq u-1$, and let $\6C_{i,j}$ and $\6C'_{i,j}$ be the $2^\ell\times u2^\ell$ matrices whose $i$-th and $j$-th blocks are $(\6M_j,\6M_i)$ and $(\6M_i,\6M_j)$, respectively, with all the other blocks equal to zero. Since dyadic matrices are symmetric and commute, and $\6M_i^2=\6M_j^2=\6I_{2^\ell}$ by Lemma~\ref{lem:rank}, we have
\begin{align*}
\6H\6C_{i,j}^T &= \6M_i\6M_j+\6M_j\6M_i=\60,\\
\6H\6C_{i,j}'^T &= \6M_i^2+\6M_j^2=\60.
\end{align*}
Hence, the rows of $\6C_{i,j}$ and $\6C'_{i,j}$ are codewords of $C$, each of weight $2v$. The rows of each matrix are distinct, since $\6M_i$ and $\6M_j$ are invertible. A row of $\6C_{i,j}$ cannot coincide with a row of $\6C'_{i,j}$, since their $i$-th blocks are a row of $\6M_j$ and a row of $\6M_i$, respectively, which differ by hypothesis. Finally, rows associated with different pairs $(i,j)$ have different nonzero blocks. Therefore, the $\binom{u}{2}$ pairs yield $\binom{u}{2}2^{\ell+1}$ distinct codewords of weight $2v$.

For $(iii)$, Lemma~\ref{lem:rank} gives $\6H\6H^T=\sum_{i=0}^{u-1}\6M_i^2=u\,\6I_{2^\ell}$, which is zero when $u$ is even. Hence, the rows of $\6H$ belong to $C$, i.e., $C^\perp\subseteq C$.
\end{IEEEproof}

The rows of a dyadic matrix are obtained by XOR-translating the support of its signature; hence, two blocks have a row in common if and only if the support of one signature is a translate of the other. This is excluded, in particular, when $v>1$ and the difference sets of the signatures are pairwise disjoint, as enforced by the heuristic of Section~\ref{subsec:heuristic}, since translated supports have the same difference set.

\begin{remark}
\label{rem:gm}
    The code described in Theorem~\ref{theo:oneRow} admits a sparse \ac{GM} with $\lambda_{\emph{r}} = 2v$. Indeed, stacking the matrices $\6C_{i,u-1}$ defined in the proof of Theorem~\ref{theo:oneRow}, for $i=0,\dots,u-2$, gives a matrix whose restriction to the first $u-1$ blocks is $\6I_{u-1}\otimes\6M_{u-1}$, which is invertible. Hence, its $2^\ell(u-1)=k$ rows are linearly independent codewords of weight $2v$ and form a basis of $C$.  
    Moreover, a \ac{GM} in systematic form can be obtained directly from the \ac{PCM} $\6H$ in systematic form, given by $\6M_{u-1}\6H$, that is, 
    \[
    \6G =
    \begin{bmatrix}[c|c]
         & \6M_{u-1} \6M_0 \\
         \6I_{2^{\ell}(u-1)} & \vdots \\
         & \6M_{u-1} \6M_{u-2}
    \end{bmatrix}.
    \]
    We observe that each row of $\6G$ has weight at most $1+v^2$. Thus, for fixed $v$, its row weight is bounded independent of the block length.  
    In the case in which $u = 2$, the associated classical code $C$ is SD, so $\6H$ itself can be chosen as a \ac{GM}.
    However, a classical SD code gives a \ac{CSS} code with $k = 0$, which therefore does not encode any logical qubit.
\end{remark}

\begin{remark}
    Claim $(iii)$ of Theorem~\ref{theo:oneRow} is also valid for even values of $v$. However, in order to get full-rank \acp{PCM}, we stick to odd values of $v$. Moreover, we use dyadic matrices with $v > 1$.
\end{remark}

Let $\6H \in \mathbb{F}_2^{2^\ell \times u2^\ell}$ be a  full-rank \ac{PCM}, representing a classical code $C$, as in~\eqref{eq:con_B}. 
The code $C$ is regular both in the column and row weight, that are $\lambda_{\text{c}} = v$ (which we always choose to be odd), and $\lambda_{\text{r}} = uv$ (even), respectively. 
The dimension and rate of $C$ are $k = (u-1)2^{\ell}$, and $R = 1 - 1/u$, respectively.
So, for the associated \ac{dc} \ac{CSS} code we get $2k - n = (u - 2)2^{\ell}$, and $uv$ as stabilizer generator weight. 

Taking into account (\ref{eq:con_B}), \emph{Construction~B} is straightforward. 
\begin{itemize}
    \item[$(i)$] We fix the parameters $u$ (even) and $v$ (odd). 
    \item[$(ii)$] Then, following the heuristic algorithm described in Section~\ref{subsec:heuristic}, we choose $u$ signatures, each associated to a different dyadic matrix $\6M_i, \forall \, i \in \{ 0, \dots, u-1 \}$.
\end{itemize} 
There are $\binom{2^\ell}{v}$ distinct signatures of weight $v$; hence, $u$ distinct dyadic blocks can be selected provided that $u\leq \binom{2^\ell}{v}$. Practical choices of the code parameters always fulfill this requirement. 
All classical and \ac{dc} \ac{CSS} codes obtained with the above procedure are referred to as \emph{Construction~B} codes.

\begin{corollary}\label{cor:min_dist_B}
Let $C$ be a \emph{Construction~B} code with even $u\geq4$ and odd signature weight $v$, and let $\mathcal C$ be the associated \ac{dc} \ac{CSS} code. Then
\[
d(C)\leq d(\mathcal C)\leq2v.
\]
\end{corollary}

\begin{IEEEproof}
   Let $\6x$ be any row of the matrix $\6C_{0,1}$ introduced in the proof of Theorem~\ref{theo:oneRow}, i.e., the matrix whose first two blocks are $\6M_1$ and $\6M_0$ and whose remaining blocks are zero. As shown there, $\6x\in C$ and $|\6x|=2v$. Suppose that $\6x\in C^\perp$, so that $\6x=\6a\6H$ for some row vector $\6a$. Since $u\geq4$, the third block of $\6x$ exists and is zero, giving $\6a\6M_2=\60$. Since $v$ is odd, $\6M_2$ is invertible by Lemma~\ref{lem:rank}; therefore, $\6a=\60$ and hence $\6x=\60$, contradicting $|\6x|=2v$. Thus, $\6x\in C\setminus C^\perp$ represents a nontrivial logical operator of weight $2v$, which proves $d(\mathcal C)\leq 2v$. The inequality $d(C)\leq d(\mathcal C)$ follows from the definition of quantum distance.
\end{IEEEproof}

Corollary~\ref{cor:min_dist_B} bounds both the classical and quantum distances of \emph{Construction~B}. For fixed $v$, the quantum distance remains bounded by $2v$ as the lifting size increases. Increasing $v$ raises this upper bound, at the cost of increasing the stabilizer generator weight $uv$.

Table~\ref{tab:constrB_distance} reports examples with $u=4$ and quantum rate $R_{\mathrm Q}=0.50$. The quantum-distance estimates obtained using~\cite{Pryadko_2022} are $\widehat d(\mathcal C)=6$, $10$, and $14$, respectively. These values coincide with the corresponding upper bounds $2v$: the search found no nontrivial logical operator of smaller weight, although this does not certify that $d(\mathcal C)=2v$. Further examples with fixed dyadic-block size and different signature weights are reported in Table~\ref{tab:code_parameters_B_4}.

\begin{table}[t]
\centering
\caption{Finite-length examples of Construction~B codes with fixed quantum rate and increasing dyadic-block size and 
weight.}
\label{tab:constrB_distance}
\begin{tabular}{|c|c|c|c|c|c|}
\hline
\textbf{Code} & \multicolumn{5}{c|}{\textbf{Parameters}} \\
\cline{2-6} 
 & $\ell$ & $u$ & $v$ & $R_{\text{Q}}$
 & $\widehat d(\mathcal C)$\hatstrut \\
\hline
$\mathcal{C}_{\text{B}, 3} \llbracket 256, 128 \rrbracket$  & $6$ & $4$ & $3$ & $0.50$ & $6$  \\
\hline
$\mathcal{C}_{\text{B}, 5} \llbracket 512, 256 \rrbracket$  & $7$ & $4$ & $5$ & $0.50$ & $10$ \\
\hline
$\mathcal{C}_{\text{B}, 7} \llbracket 1024, 512 \rrbracket$  & $8$ & $4$ & $7$ & $0.50$ & $14$ \\
\hline
\end{tabular}
\end{table}

\begin{remark}
  \emph{Construction~B} is a generalization of the well-known family of \ac{GB} codes. 
\end{remark}

Next, we provide a heuristic algorithm that improves the minimum distance and reduces the number of length-$4$ cycles in the Tanner graph of \emph{Construction~B} codes.

\subsection{Heuristic Optimization of \emph{Construction~B} Codes} \label{subsec:heuristic}

For fixed $\ell$, $u$, and $v$, the dyadic structure forces at least $u\binom{v}{2}2^{\ell-1}$ length-$4$ cycles, as established below. Additional cycles arise from repeated pairwise XOR values within a signature or across different signatures. The heuristic aims to eliminate these additional cycles, leaving only the unavoidable contribution. To describe cycles spanning two distinct blocks, consider $\6M_i$ and $\6M_{i'}$, with $i\neq i'$. A cycle occurs whenever their signature supports, $\supp(\6m_i)$ and $\supp(\6m_{i'})$, contain pairs of indices whose binary representations generate the same bitwise XOR. More precisely, if there exist distinct indices $a,b\in\supp(\6m_i)$ and distinct indices $a',b'\in\supp(\6m_{i'})$ such that $a\oplus b=a'\oplus b'$, then the Tanner graph contains a length-$4$ cycle involving columns from both blocks. More insights on cycle properties of \ac{QD} codes are given in Section~\ref{sec:cycles}.

In the following, we introduce a heuristic algorithm that aims at minimizing the number of length-4 cycles.
Let $2^{\ell}$ be the size of each block $\6M_i$, $v$ the common signature weight of such blocks, with $2^\ell > v$, and $u$ the number of dyadic blocks. Moreover, let $A = \{0, 1, \dots, 2^{\ell}-1 \}$.
Our aim is to generate the $u$ supports of the signatures such that, for $i, j_1, j_2 \in \{0, 1, \dots, u-1 \}$:
\begin{itemize}
    \item[$(i)$] the entries of each $\supp(\6m_i)$ are uniformly distributed;
    \item[$(ii)$] for each $\supp(\6m_i)$, the pairwise XORs of the binary representation of its elements are mutually distinct. 
    The resulting collection of values defines the \emph{difference set} $W_i$, whose cardinality is $\binom{v}{2}$;
    \item[$(iii)$] the $u$ difference sets $W_i$ are pairwise disjoint; that is, for any $j_1 \neq j_2$, \(W_{j_1} \cap W_{j_2} = \emptyset\).
\end{itemize}

\begin{table}[t!]
\centering
\caption{Parameters of classical and \ac{CSS} codes designed through the proposed constructions. 
For Construction~A, the design parameters are reported.  
$U(d)$ denotes the upper bound on the minimum distance of the underlying classical code.}
\label{tab:code_parameters}
\setlength{\tabcolsep}{3.2pt}
\renewcommand{\arraystretch}{1.15}
\resizebox{\columnwidth}{!}{%
\begin{tabular}{|c|c|c|c|c|c|c|c|}
\hline
\textbf{Const.} & \multicolumn{4}{c|}{\textbf{Classical Code}} & \multicolumn{3}{c|}{\textbf{DC CSS Code}} \\
\cline{2-8}
& $k$ & $n$ & $R$ & $U(d)$ & $k_{\text{Q}}$ & $n_{\text{Q}}$ & $R_{\text{Q}}$ \\
\hline
A 
& $2^\ell (u-w)$  
& $2^\ell u$ 
& $\frac{u-w}{u}$ 
& $u$ 
& $2^\ell (u - 2w)$ 
& $2^\ell u$ 
& $\frac{u - 2w}{u}$ \\
\hline
B 
& $2^\ell (u-1)$ 
& $2^\ell u$ 
& $\frac{u-1}{u}$ 
& $2v$ 
& $2^\ell (u - 2)$ 
& $2^\ell u$ 
& $\frac{u - 2}{u}$ \\
\hline
\end{tabular}%
}
\end{table}

For a single dyadic matrix of weight~$v$, $n_4\ge\binom{v}{2}2^{\ell-1}$~\cite[Rem.~4.17]{kirsten_QD}. The following proposition extends this bound to \emph{Construction~B}.
\begin{proposition}
\label{prop:n4_B}
    Let $\6H$ be the \ac{PCM} of a \emph{Construction~B} code, as in~\eqref{eq:con_B}. Then its Tanner graph contains $n_4\ge u\binom{v}{2}2^{\ell-1}$ length-$4$ cycles, with equality if and only if conditions~$(ii)$ and~$(iii)$ are satisfied.
\end{proposition}

\begin{IEEEproof}
    For $\delta\neq 0$, let $k_\delta$ be the number of pairs $\{a,b\}\subseteq\supp(\6m_i)$, over all $i$, with $a\oplus b=\delta$, so that $\sum_{\delta\neq 0}k_\delta=u\binom{v}{2}$. Since row $r$ of $\6M_i$ has support $\supp(\6m_i)\oplus r$, where $\6m_i$ is the signature of the $i$-th dyadic matrix, each such pair gives rows $r$ and $r\oplus\delta$ two common positions, $a\oplus r$ and $b\oplus r$, for all the $2^{\ell-1}$ pairs of rows at XOR distance~$\delta$. Thus, $n_4=2^{\ell-1}\sum_{\delta\neq 0}\binom{2k_\delta}{2}\ge 2^{\ell-1}\sum_{\delta\neq 0}k_\delta=u\binom{v}{2}2^{\ell-1}$, since $\binom{2k}{2}=k(2k-1)\ge k$, with equality if and only if all $k_\delta\le 1$, i.e., if and only if $(ii)$ and $(iii)$ hold.
\end{IEEEproof}
Thus, the dyadic structure forces at least $\binom{v}{2}2^{\ell-1}$ length-$4$ cycles per block. When conditions~$(ii)$ and~$(iii)$ hold, the Tanner graph contains exactly $u\binom{v}{2}2^{\ell-1}$ length-$4$ cycles, with no additional contribution from repeated XOR values.

In order to avoid the situation in which a large number of elements of the arbitrary support $\supp(\6m_i)$ fall into a small subset of $A$, we fix $m \le \ell$ as the smallest integer such that $2^{m} > v$, and we split $A$ into the $2^{m}$ subsets $J_{k}=\{kq,\dots,(k+1)q-1\}$, with $q=2^{\ell-m}$ and $k = 0, \dots, 2^{m}-1$.
Then, for each $i \in \{0, \dots, u-1\}$, we randomly select $v$ distinct subsets $J_k$, $\lceil v/2\rceil$ from one half of $A$ and $\lfloor v/2\rfloor$ from the other, alternating the two halves with the parity of $i$, and we choose one element of $\supp(\6m_i)$ from each selected subset. A candidate element is accepted only if its XORs with the elements already selected for $\supp(\6m_i)$ are new, i.e., they appear neither in $W_i$ nor in the difference sets of the previous supports. If no valid candidate is found in a subset within $\mathtt{th}$ trials, the subsets are selected again, up to $\mathtt{MaxAttempts}$ times.
A pseudo-code description of this heuristic procedure is provided in Algorithm~\ref{alg:heuristic}, where $\mathtt{Random}(S,k)$ returns $k$ distinct elements chosen uniformly at random from the set $S$. Moreover, $x\xleftarrow{\$}S$ denotes a uniformly random choice of $x$ in $S$, $\land$ and $\neg$ denote the logical AND and NOT operators, and $\bot$ indicates that no valid support was found within $\mathtt{MaxAttempts}$ attempts. Upon successful termination, the algorithm returns a matrix $\6F\in\{0,\dots,2^\ell-1\}^{u\times v}$, whose $i$-th row contains the $v$ elements of the support $\supp(\6m_i)$.

\begin{algorithm}[t]
\footnotesize
\KwIn{$\ell$, $u$, odd signature weight $v$, $\mathtt{MaxAttempts}\in\mathbb N$, threshold $\mathtt{th}\in\mathbb N$}
\KwOut{$\6F\in\{0,\dots,2^{\ell}-1\}^{u\times v}$, or $\bot$}
$m\gets\lfloor\log_2 v\rfloor+1$, \ $r\gets 2^{m}$, \ $q\gets 2^{\ell-m}$\;
$B\gets\emptyset$, \ $\6F\gets\60_{u\times v}$;\tcp{$B$: differences already used}
\For{$i=0,\dots,u-1$}{
  $\mathtt{success}\gets\text{false}$, \ $\mathtt{Attempts}\gets 0$\;
  \While{$\neg\,\mathtt{success}\land\mathtt{Attempts}<\mathtt{MaxAttempts}$}{
    $(f,s)\gets(\lceil v/2\rceil,\lfloor v/2\rfloor)$ if $i$ is odd, else $(\lfloor v/2\rfloor,\lceil v/2\rceil)$\;
    $\mathcal K\gets\mathtt{Random}(\{0,\dots,r/2-1\},f)\cup\mathtt{Random}(\{r/2,\dots,r-1\},s)$\;
    $W\gets\emptyset$, \ $M\gets\emptyset$;\tcp{current $W_i$ and $\supp(\6m_i)$}
    $\mathtt{failed}\gets\text{false}$\;
    \For{$j\in\mathcal K$}{
      $\mathtt{found}\gets\text{false}$, \ $t\gets 0$\;
      \While{$t<\mathtt{th}\land\neg\,\mathtt{found}$}{
        $x\xleftarrow{\$}J_j$, \ $t\gets t+1$;\tcp{candidate from $J_j$}
        $N\gets\{x\oplus y: y\in M\}$;\tcp{new XORs}
        \lIf{$N\cap(W\cup B)=\emptyset$}{$M\gets M\cup\{x\}$, \ $W\gets W\cup N$, \ $\mathtt{found}\gets\text{true}$}
      }
      \lIf{$\neg\,\mathtt{found}$}{$\mathtt{failed}\gets\text{true}$; \textbf{break}}
    }
    \lIf{$\neg\,\mathtt{failed}$}{$B\gets B\cup W$, \ $\6F(i,:)\gets\mathtt{sorted}(M)$, \ $\mathtt{success}\gets\text{true}$}
    $\mathtt{Attempts}\gets\mathtt{Attempts}+1$\;
  }
  \lIf{$\neg\,\mathtt{success}$}{\Return $\bot$}
}
\Return $\6F$\;
\caption{Heuristic algorithm for \emph{Construction~B} codes}\label{alg:heuristic}
\end{algorithm}

\subsection{Summary of \emph{Constructions A} and \emph{B}}
The design parameters of the classical and \ac{CSS} codes obtainable through the two above constructions are summarized in Table~\ref{tab:code_parameters}, as a function of $u$, $v$, $w$, and $\ell$. By Theorem~\ref{theo:girt_DC}, the \ac{PCM} of a \ac{dc} code with column weight larger than $1$ always contains length-$4$ cycles; hence, all the \emph{Construction~A} and \emph{B} codes considered in this paper have girth $4$.
Whenever Algorithm~\ref{alg:heuristic} returns a matrix, the corresponding signature supports satisfy conditions~$(ii)$ and~$(iii)$. By Proposition~\ref{prop:n4_B}, the resulting Tanner graph therefore contains exactly $u\binom{v}{2}2^{\ell-1}$ length-$4$ cycles, attaining the minimum possible number for the chosen parameters.

\section{Cycle Properties of \ac{QD} Codes} \label{sec:cycles}

In the following, given $\6D^{(i)} \in \mathcal D_\ell(\mathbb F_2)$, we denote the binary representation of the position $i$ of the unique non-zero element in its signature as $\boldsymbol{\rho}=\phi(i)\in\mathbb F_2^\ell$, where $\phi$ has been defined in Section~\ref{sec:dyadic}.
\ac{QD} matrices constructed from dyadic blocks whose signatures have weight greater than one always exhibit girth~$4$, as shown in~\cite{dyadics, Martinez2022}.
We thus restrict our attention to \ac{QD} matrices built from \acp{DPM}. 
Next, two theoretical results from~\cite{dyadics, Martinez2022} are reported, particularized to the dyadic case as follows. 
\begin{lemma}\label{lemma:kelley}
    If $\6M = 
    \begin{bmatrix}
        \6D_{0,0} & \6D_{0,1} \\
        \6D_{1,0} & \6D_{1,1}
    \end{bmatrix}$ 
    where $\6D_{i, j}$ is a \ac{DPM} with the non-zero signature entry in position  $\boldsymbol{\rho}_{i,j} \in \mathbb{F}^{\ell}_{2}$, then each cycle in the \ac{QD} matrix $\6M$ has the same length.
\end{lemma}
\begin{theorem} \label{theo:kelley}
    If $\6M = 
    \begin{bmatrix}
        \6D_{0,0} & \6D_{0,1} \\
        \6D_{1,0} & \6D_{1,1}
    \end{bmatrix}$ 
    where $\6D_{i, j}$ is a \ac{DPM} with the non-zero signature entry in position $\boldsymbol{\rho}_{i, j} \in \mathbb{F}^{\ell}_{2}$, then the girth of the 
    \ac{QD} matrix $\6M$ is 
    \begin{equation} \label{eq:girthorder}
        g = 4(\#\langle \boldsymbol\alpha \rangle),
    \end{equation} 
    where $\boldsymbol\alpha = (\boldsymbol{\rho}_{0, 0} \oplus \boldsymbol{\rho}_{1, 1}) \oplus (\boldsymbol{\rho}_{0, 1} \oplus \boldsymbol{\rho}_{1, 0})$.
\end{theorem}

\begin{remark} \label{rem:cycles}
    Since the order of an element in $\mathbb{F}^{\ell}_{2}$ is always $1$ (if it is null) or $2$ (if not), from Lemma~\ref{lemma:kelley} and Theorem~\ref{theo:kelley}, we conclude that for any $2 \times 2$ array of \acp{DPM}, each cycle has length $8$, in the best case, or $4$, in the worst case. So, the girth of any \ac{QD} $2 \times 2$ array of \acp{DPM} is always $4$ or $8$.
\end{remark}

Let us now focus on larger \ac{QD} constructions. 

\begin{corollary}
    Let 
    \begin{equation} \label{eq:formM}
        \6M = 
        \begin{bmatrix}
            \6D_{0, 0} &  \6D_{0, 1} & \cdots & \6D_{0, u-1} \\
            \vdots & \vdots & \ddots & \vdots \\
            \6D_{w-1, 0} &  \6D_{w-1, 1} & \cdots & \6D_{w-1, u-1}
        \end{bmatrix},
    \end{equation}
    be a \ac{QD} matrix, where $\6D_{i, j} \in \mathcal{D}_{\ell}(\mathbb F_2)$ and $u,w\geq 2$.  
    If every possible $2\times 2$ sub-array of $\6M$, i.e.,
    \begin{equation} \label{eq:square}
        \6M' = 
        \begin{bmatrix}
            \6D_{i, j} & \6D_{i, j'} \\
            \6D_{i', j} & \6D_{i', j'}
        \end{bmatrix},
    \end{equation}
    with $i \neq i'$ and $j \neq j'$ has girth equal to $8$, then the girth of the code associated to  $\6M$ is either $6$ or $8$.
    \label{cor:no4}
\end{corollary}
\begin{IEEEproof}
A length-$4$ cycle involves two block-rows and two block-columns of $\6M$, hence it lies in one of the sub-arrays~\eqref{eq:square}. By Theorem~\ref{theo:kelley}, each such sub-array has girth $4$ if $\boldsymbol\alpha=\60$ and $8$ otherwise. Hence, by hypothesis, $\6M$ has no length-$4$ cycles, i.e., $g\geq6$; moreover, $g\leq8$, since every sub-array~\eqref{eq:square} contains length-$8$ cycles by Lemma~\ref{lemma:kelley}.
\end{IEEEproof}

\begin{lemma}
\label{lem:g6}
Consider a $3 \times 3$ \ac{QD} matrix $\mathbf{M}$ of the form
\begin{equation}
  \mathbf{M} = 
  \begin{bmatrix}
    \6D_{0,0} & \60   & \6D_{0,2} \\
    \6D_{1,0} & \6D_{1,1} & \60 \\
    \60 & \6D_{2,1} & \6D_{2,2}
  \end{bmatrix},
  \label{eq:cyc6}
\end{equation}
where each $\6D_{i,j} \in \mathcal{D}_{\ell}(\mathbb{F}_2)$.  
Let
\[
\boldsymbol\alpha
=
\bigoplus_{\substack{(i,j) \\ \6D_{i,j} \neq \6{0}}}
\boldsymbol{\rho}_{i,j}.
\]
Then a cycle of length $6$ exists in  
$\mathbf{M}$ if and only if \(
\#\langle \boldsymbol\alpha \rangle = 1.
\)
Moreover, the same condition holds for any matrix obtained from $\mathbf{M}$ through arbitrary row and/or column permutations.
\end{lemma}

\begin{IEEEproof}
Starting from a variable node in block-column $0$ with internal
index $\mathbf x\in\mathbb F_2^\ell$, traverse the nonzero
blocks in the order \((0,0),\ (0,2),\ (2,2),\ (2,1),\ (1,1),\ (1,0).\) The traversal returns to block-column $0$ with internal index
$\mathbf x\oplus\boldsymbol\alpha$. Thus, it closes if and
only if $\boldsymbol\alpha=\mathbf0$. When it closes, it is
a simple length-$6$ cycle, since it visits each of the three
block-rows and each of the three block-columns once.

Conversely, because each nonzero block is a permutation
matrix, any length-$6$ cycle in this matrix must traverse
these six blocks in cyclic order. Its closure therefore
requires $\boldsymbol\alpha=\mathbf0$, equivalently
$\#\langle\boldsymbol\alpha\rangle=1$.

Arbitrary row and column permutations only relabel the
check and variable nodes of the Tanner graph. They preserve
cycle lengths, so the permuted matrix contains a length-$6$
cycle if and only if the original matrix does.
\end{IEEEproof}

We now provide a girth criterion for a generic $w\times u$ array of \acp{DPM}, following the formulation in~\cite[Theorem~2.1]{Fossorier2004} for arrays of circulant permutation matrices.

\begin{theorem}
Let $\mathbf M$ be a $w\times u$ array of \acp{DPM}, and define
\[
\boldsymbol{\Delta}_{a,b}(l)
=\boldsymbol{\rho}_{a,l}\oplus\boldsymbol{\rho}_{b,l}.
\]
For any integer $t\geq2$, the Tanner graph of $\mathbf M$ has girth $g\geq2(t+1)$ if and only if
\begin{equation}\label{eq:fossoqd}
\bigoplus_{k=0}^{\mu-1}\boldsymbol{\Delta}_{j_k,j_{k+1}}(l_k)
\neq\mathbf0
\end{equation}
for every $\mu\in\{2,\ldots,t\}$, all block-index sequences $j_0,\ldots,j_\mu\in\{0,\ldots,w-1\}$ and $l_0,\ldots,l_{\mu-1}\in\{0,\ldots,u-1\}$ satisfying
\[
j_\mu=j_0, \quad
j_k\neq j_{k+1}, \qquad
l_\mu=l_0, \quad
l_k\neq l_{k+1},
\]
for $k=0,\ldots,\mu-1$. 
\end{theorem}

\begin{IEEEproof}
Fix an admissible block-index sequence and a starting check node with internal index $\mathbf x\in\mathbb F_2^\ell$. Traversing the two edges through block-column $l_k$, from block-row $j_k$ to block-row $j_{k+1}$, changes the internal index by $\boldsymbol{\Delta}_{j_k,j_{k+1}}(l_k)$. After $2\mu$ edges, the internal index is therefore
\begin{equation}
\label{eq:theo5}
    \mathbf x \oplus \Bigg( \bigoplus_{k=0}^{\mu-1} \boldsymbol{\Delta}_{j_k,j_{k+1}}(l_k) \Bigg).
\end{equation}
Hence, the walk closes if and only if~\eqref{eq:theo5} is equal to zero. 
The inequalities on adjacent block indices prevent immediate reversal along an edge.

If~\eqref{eq:fossoqd} fails for some $\mu\leq t$, the resulting closed walk contains a simple cycle of length at most $2\mu$. Then, $g\leq2t$.
Conversely, any simple cycle of length $2\mu\leq2t$ determines an admissible block-index sequence. Since the cycle closes, its cumulative XOR is zero, violating~\eqref{eq:fossoqd}. Thus, requiring~\eqref{eq:fossoqd} for every admissible sequence and every $\mu\in\{2,\ldots,t\}$ is equivalent to excluding all cycles of length at most $2t$. The Tanner graph is bipartite, so this is equivalent to $g\geq2(t+1)$.
\end{IEEEproof}

\section{Numerical Results}
\label{sec:results}

In this section, we show the \ac{LER} of the proposed \ac{QLDPC} codes through finite-length numerical simulations. 
We consider both depolarizing and phenomenological-noise settings\footnote{The \acp{PCM} of all the codes considered in this section can be found at $\text{\url{https://github.com/secomms/DC_QD_Codes}}$.}.

\subsection{Performance on Depolarizing Noise}
\label{subsec:depo_noise}

We evaluate the \ac{LER} of some \ac{CSS} \emph{Construction A} and \emph{B} codes, in both the short and moderate block-length regime, by Monte Carlo simulations on depolarizing noise. 
Depending on the code length and on the purpose of each experiment, we use one of the following decoders: 
\begin{itemize}
    \item a \ac{BP4} decoder~\cite{Babar2015};
    \item a \ac{BP4-M} decoder~\cite{Babar2015}, effective for \ac{dc} \ac{CSS} codes;
    \item a \ac{BP} decoder, \ac{MS} algorithm, with ordered-statistics decoding (OSD) post-processing, with $\mathtt{order}=10$, referred to as \acs{BP+OSD}. We use the decoder implementation provided in~\cite{Roffe_decoding}. We decode using $\6H_1$, and report the corresponding \ac{LER}. We remark that, for \ac{dc} codes, the two \acp{PCM} are identical.
\end{itemize}
All decoders run a flooding schedule  BP message passing for at most $I_{\text{max}} = 100$ iterations, and, for each value of $p$, that is, the depolarizing probability, the simulation continues until $100$ logical errors are collected.
For each \ac{CSS} \emph{Construction A} and \emph{B} code, we compare the \ac{LER} performance against instances with similar block lengths, stabilizer generator weights and rates\footnote{For each code family, we select from the literature the best-performing code instance with comparable block length, rate and stabilizer generator weight. When no suitable instance is available, we construct one, whenever possible, by carefully optimizing its performance.}, from several state-of-the-art \ac{QLDPC} code construction methods, namely\footnote{We do not include comparisons with \ac{LP} and \ac{BB} codes, since the representative instances available at comparable block lengths and stabilizer generator weights have substantially lower code rates, preventing a fair comparison.}:
\begin{itemize}
    \item bicycle codes, constructed as in~\cite{sparse_quantum_McKay}\footnote{For each set of parameters, we generate $20$ different code instances, evaluate their decoding performance, and retain the best-performing one.} (in Table~\ref{tab:code_parameters_bic_codes});
    \item \Ac{QC} codes, constructed as in~\cite{Hagiwara2007} (in Table~\ref{tab:code_parameters_QC_codes});
    \item \ac{HP} codes~\cite{Tillich_HGP} (in Table~\ref{tab:code_parameters_HP_codes});
    \item \ac{GB} codes~\cite{kovalev_Pryad_GBcodes} (in Table~\ref{tab:code_parameters_GB_codes}), constructed as in~\cite{Panteleev_degenerate, GB_codes_Lin};
    \item \ac{QT} codes~\cite{leverrier2022quantumtannercodes, Olai_GQT, Asynt_Lin} (in Table~\ref{tab:code_parameters_QT_codes}).
\end{itemize}
Tables~\ref{tab:code_parameters_bic_codes}--\ref{tab:code_parameters_QT_codes} report the row and column weights, minimum distance, girth, number of length-$4$ cycles, and quantum rate of each tested code. 
For non-\ac{dc} codes, $n_{4,1}$ and $n_{4,2}$ refer to the number of length-$4$ cycles of $C_1$ and $C_2^\perp$, respectively; for \ac{dc} codes, we report their common value $n_4$. A missing reference, denoted with ``$\cdot$'',  indicates a code instance that is not available in the literature and is thus constructed by the authors. 
As in Section~\ref{sec:QD_LDPC}, $\widehat d(\mathcal C)$ denotes the smallest weight of a nontrivial logical operator found with~\cite{Pryadko_2022}, so that $d(\mathcal C)\le\widehat d(\mathcal C)$; when the distance is specified in the cited reference, $\widehat d(\mathcal C)=d(\mathcal C)$.

More specifically, Table~\ref{tab:code_parameters_QC_codes} summarizes the parameters $J = K$, $L$, $P$, $\sigma$, and $\tau$ according to the notation of~\cite{Hagiwara2007}. 
A \ac{CSS} \ac{QC} $\llbracket n, k \rrbracket$ code is denoted by $\mathcal{C}_{\text{QC}}$, with $\lambda_{\text{r}} = L$ and $\lambda_{\text{c}} = J = K$. 
The parameters $P$, $\sigma$, and $\tau$ are reported for reproducibility. 
The number of length-4 cycles is not reported since the girth of these codes is $6$.
This class of \ac{CSS} codes can be considered as a specific case of \ac{GB} codes, as reported in~\cite{Pant_Kal_almost_linear}.  
In Table~\ref{tab:code_parameters_HP_codes}, the parameters for \ac{HP} codes are reported. 
We construct these codes using a single classical seed code, as in~\cite{Roffe_decoding, Ostrev2024classicalproduct}, which we have chosen by maximizing its minimum distance, in accordance with the bounds available in~\cite{Grassl_table}. 

\begin{table}[tb!]
\centering
\caption{Parameters of tested \ac{CSS} Bicycle codes.
}
\begin{tabular}{|l|c|c|c|c|c|c|}
\hline
\textbf{Code} & \multicolumn{6}{c|}{\textbf{Parameters}} \\
\cline{2-7}
                       & $\lambda_{\text{r}}$ & $\lambda_{ \text{c, avg}}$ & $\widehat{d}(\mathcal{C})$\hatstrut & $g$ & $n_4$ & $R_{\text{Q}}$ \\
\hline
$\mathcal{C}_{\text{Bic}} \llbracket 128, 24 \rrbracket$ & $8$ & $3.25$ & $4$ & $4$ & $338$ & $0.19$ \\  
\hline
$\mathcal{C}_{\text{Bic}} \llbracket 512, 160 \rrbracket$ & $8$ & $2.75$ & $2$ & $4$ & $815$ & $0.31$ \\ 
\hline
$\mathcal{C}_{\text{Bic}} \llbracket 128, 64 \rrbracket$ & $12$ & $3.0$ & $2$ & $4$ & $453$ & $0.50$ \\ 
\hline
$\mathcal{C}_{\text{Bic}} \llbracket 512, 256 \rrbracket$ & $12$ & $3.0$ & $2$ & $4$ & $1\,026$ & $0.50$ \\ 
\hline
\end{tabular}
\label{tab:code_parameters_bic_codes}
\end{table}

\begin{table}[tb!]
\centering
\caption{Parameters of tested \ac{CSS} \ac{QC} codes. 
}
\begin{tabular}{|l|c|c|c|c|c|c|c|c|}
\hline
\textbf{Code} & \multicolumn{8}{c|}{\textbf{Parameters}} \\
\cline{2-9}
 & $J$ & $L$ & $P$ & $\sigma$ & $\tau$ & $\widehat{d}(\mathcal{C})$\hatstrut & $g$ & $R_{\text{Q}}$ \\
\hline
$\mathcal{C}_{\text{QC}} \llbracket 136, 38 \rrbracket$ & $3$ & $8$ & $17$ & $4$ & $7$ & $6$ & $6$ & $0.28$ \\
\hline
$\mathcal{C}_{\text{QC}} \llbracket  156, 82 \rrbracket$ & $3$ & $12$ & $13$ & $4$ & $7$ & $6$ & $6$ & $0.53$ \\
\hline
$\mathcal{C}_{\text{QC}} \llbracket 516, 262 \rrbracket$ & $3$ & $12$ & $43$ & $7$ & $5$ & $8$ & $6$ & $0.51$ \\
\hline
\end{tabular}
\label{tab:code_parameters_QC_codes}
\end{table}

\begin{table*}[tb!]
\centering
\caption{Parameters of tested \ac{CSS} \ac{HP} codes. 
\\ With “$\cdot$” in the “\textbf{Ref.}” column we indicate that we have constructed the corresponding code instance.
}
\begin{tabular}{|l|c|c|c|c|c|c|c|c|}
\hline
\textbf{Code} & {\textbf{Seed code}} & \multicolumn{6}{c|}{\textbf{Parameters}} & {\textbf{Ref.}} \\
\cline{3-8}
         & & $\lambda_{\text{r, avg}}$ & $\lambda_{ \text{c, avg}}$ & $\widehat{d}(\mathcal{C})$\hatstrut & $g$ & $n_{4,1} / n_{4, 2}$ & $R_{\text{Q}}$ & \\
\hline
$\mathcal{C}_{\text{HP}} \llbracket 125, 25 \rrbracket$ & $C[10, 5, 4]$ &  $6.60$ & $2.64$ & $4$ & $4$ & $225 / 225$ & $0.20$ & $\cdot$ \\ 
\hline
$\mathcal{C}_{\text{HP}} \llbracket 130, 64 \rrbracket$ & $C[11, 8, 2]$ &  $7.21$ & $1.83$ & $2$ & $4$ & $70 / 70$ & $0.49$ & $\cdot$ \\ 
\hline
$\mathcal{C}_{\text{HP}} \llbracket 505, 169 \rrbracket$ & $C[21, 13, 4]$ & $10.88$ & $3.62$ & $4$ & $4$ & $1\,653 / 1\,653$ & $0.33$ & \cite{Ostrev2024classicalproduct} \\
\hline
$\mathcal{C}_{\text{HP}} \llbracket 520, 256 \rrbracket$ & $C[22, 16, 4]$ & $11.88$ & $3.02$ & $4$ & $4$ & $2\,464 / 2\,464$ & $0.49$ & $\cdot$ \\
\hline
\end{tabular}
\label{tab:code_parameters_HP_codes}
\end{table*}

\begin{table}[tb!]
\centering
\caption{Parameters of tested \ac{CSS} \ac{GB} codes.}
\label{tab:code_parameters_GB_codes}
\resizebox{\columnwidth}{!}{%
\begin{tabular}{|l|c|c|c|c|c|c|c|}
\hline
\textbf{Code} & \multicolumn{6}{c|}{\textbf{Parameters}} & \textbf{Ref.} \\
\cline{2-7}
 & $\lambda_{\text{r}}$ & $\lambda_{c,\mathrm{avg}}$ & $\widehat{d}(\mathcal{C})$\hatstrut & $g$ & $n_{4,1}/n_{4,2}$ & $R_{\text{Q}}$ & \\
\hline
$\mathcal{C}_{\mathrm{GB}} \llbracket 126, 28 \rrbracket$ 
& $10$ & $5.0$ & $8$ & $4$ & $189/189$ & $0.22$ 
& \cite{Panteleev_degenerate} \\
\hline
$\mathcal{C}_{\mathrm{GB}} \llbracket 168, 76 \rrbracket$ 
& $12$ & $3.29$ & $4$ & $4$ & $120/112$ & $0.45$ 
& \cite{nisqec_gb_mmm} \\
\hline
$\mathcal{C}_{\mathrm{GB}} \llbracket 512, 174 \rrbracket$ 
& $8$ & $2.64$ & $6$ & $6$ & $0/0$ & $0.34$ 
& \cite{nisqec_gb_mmm} \\
\hline
\end{tabular}%
}
\end{table}

\begin{table}[tb!]
\centering
\caption{Parameters of tested \ac{CSS} \ac{QT} codes.}
\label{tab:code_parameters_QT_codes}
\resizebox{\columnwidth}{!}{%
\begin{tabular}{|l|c|c|c|c|c|c|c|}
\hline
\textbf{Code} & \multicolumn{6}{c|}{\textbf{Parameters}} & \textbf{Ref.} \\
\cline{2-7}
 & $\lambda_{\text{r}}$ & $\lambda_{c,\mathrm{avg}}$ & $\widehat{d}(\mathcal{C})$\hatstrut & $g$ & $n_{4,1}/n_{4,2}$ & $R_{\text{Q}}$ & \\
\hline
$\mathcal{C}_{\mathrm{QT}} \llbracket 168, 78 \rrbracket$ 
& $12$ & $3.43$ & $3$ & $4$ & $904/968$ & $0.46$ 
& \cite{nisqec_gb_mmm} \\ 
\hline
$\mathcal{C}_{\mathrm{QT}} \llbracket 180, 26 \rrbracket$ 
& $6$ & $2.67$ & $6$ & $4$ & $130/130$ & $0.14$ 
& \cite{nisqec_gb_mmm} \\ 
\hline
$\mathcal{C}_{\mathrm{QT}} \llbracket 500, 188 \rrbracket$ 
& $10$ & $3.20$ & $4$ & $4$ & $2\,496/2\,496$ & $0.38$ 
& \cite{nisqec_gb_mmm} \\
\hline
$\mathcal{C}_{\mathrm{QT}} \llbracket 504, 223 \rrbracket$ 
& $12$ & $3.43$ & $4$ & $4$ & $2\,352 / 2\,412$ & $0.44$ 
& \cite{nisqec_gb_mmm} \\ 
\hline
\end{tabular}%
}
\end{table}

\begin{table}[tb!]
\centering
\caption{Parameters of \ac{dc} \ac{CSS} \emph{Construction~A} codes. 
}
\begin{tabular}{|l|c|c|c|c|c|c|c|}
\hline
\textbf{Code} & \multicolumn{7}{c|}{\textbf{Parameters}} \\
\cline{2-8}
  & $w$ & $u$ & $\ell$ & $\widehat{d}(\mathcal{C})$\hatstrut & $g$ & $n_4$ & $R_{\text{Q}}$ \\
\hline                    
$\mathcal{C}_{\text{A}} \llbracket 128, 24 \rrbracket$ & $5$ & $8$ & $4$ & $8$ & $4$ & $640$ & $0.19$ \\
\hline
$\mathcal{C}_{\text{A}} \llbracket 512, 160 \rrbracket$ & $4$ & $8$ & $6$ & $4$ & $4$ & $1\,536$ & $0.31$ \\
\hline
\end{tabular}
\label{tab:code_parameters_A}
\end{table}

\subsection{Performance of \emph{Construction~A} Codes}
\label{subsec:performance_A}

A summary of the values of the parameters $w$, $u$, $\ell$, $\widehat{d}(\mathcal{C})$, $g$, $n_4$ and $R_{\text{Q}}$ for the two \ac{dc} \ac{CSS} \emph{Construction~A} codes is shown in Table~\ref{tab:code_parameters_A}. 
Due to the overcompleteness of the associated \acp{PCM}, although the design rate $(u-2w)/u$ in Table~\ref{tab:code_parameters} is non-positive for both codes, the \emph{actual rate} is $R_{\text{Q}} = 0.19$ and $R_{\text{Q}} = 0.31$, for the first and second code reported in Table~\ref{tab:code_parameters_A}, respectively.
\emph{Construction~A} yields stabilizer generator weight $u = 8$, whereas non-trivial \emph{Construction~B} codes have a minimum allowable stabilizer generator weight of $12$.

\begin{figure}[t]
    \centering
    \resizebox{0.80\columnwidth}{!}{
        \input{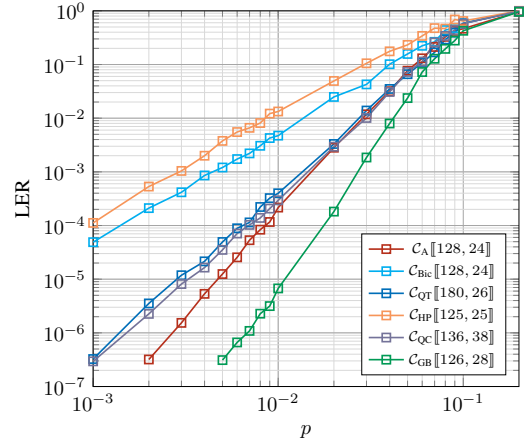}
    }
    \caption{Comparison between the \ac{LER} of $\mathcal{C}_{\text{A}}\llbracket 128, 24 \rrbracket$  and several other codes, as a function of $p$, using the \ac{BP+OSD} decoder.}
    \label{fig:constrA_[[128,24,8]]}
\end{figure}

In Fig.~\ref{fig:constrA_[[128,24,8]]}, the performance of  $\mathcal{C}_{\text{A}}\llbracket 128, 24 \rrbracket$, compared to several instances of state-of-the-art \ac{QLDPC} code constructions and using \ac{BP+OSD} decoder, is reported.
In particular, we take into account one code instance for each \ac{QLDPC} code construction, in the short block-length scenario $(n \approx 128)$, with stabilizer generator weight $\approx 8$ and $R_{\text{Q}} \approx 0.19$.
The proposed $\mathcal{C}_{\text{A}}\llbracket 128, 24 \rrbracket$ is characterized by better performance than $\mathcal{C}_{\text{Bic}}\llbracket 128, 24 \rrbracket$   and  $\mathcal{C}_{\text{HP}}\llbracket 125, 25 \rrbracket$, in line with their smaller minimum distances, namely $\widehat{d}(\mathcal{C}_{\text{Bic}}) = \widehat{d}(\mathcal{C}_{\text{HP}}) = 4$. 
Moreover, the same holds also for  $\mathcal{C}_{\text{QT}}\llbracket 180, 26 \rrbracket$  and  $\mathcal{C}_{\text{QC}}\llbracket 136, 38 \rrbracket$, since they have $\widehat{d}(\mathcal{C}_{\text{QT}}) = \widehat{d}(\mathcal{C}_{\text{QC}}) = 6$.
Finally, we also add  the performance of $\mathcal{C}_{\text{GB}}\llbracket 126, 28 \rrbracket$, called “A$2$” in~\cite{Panteleev_degenerate}, as a benchmark, since the weight of stabilizer generators is larger than that of the other codes, namely $10$. 
Although $\widehat{d}(\mathcal{C})$ and $g$ are the same as those of the \ac{QD} code, the \ac{GB} code performs better.
However, under phenomenological noise, $\mathcal{C}_{\text{A}}$ shows performance in line with that of $\mathcal{C}_{\text{GB}}$ (see Section~\ref{subsec:PLnoise}).

\begin{figure}[t]
    \centering
    \resizebox{0.80\columnwidth}{!}{
        \input{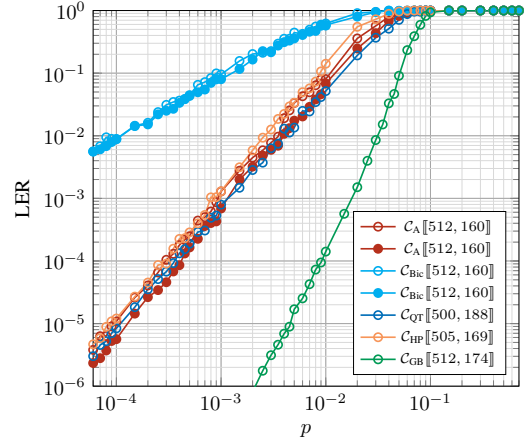}
    }
    \caption{Comparison between the \ac{LER} of $\mathcal{C}_{\text{A}}\llbracket 512, 160 \rrbracket$  and several other codes, as a function of $p$, using the \ac{BP4} (empty dot lines) and the \ac{BP4-M} (full dot lines) decoders.}
    \label{fig:constrA_[[512,160,4]]}
\end{figure}

In Fig.~\ref{fig:constrA_[[512,160,4]]}, the performance of  $\mathcal{C}_{\text{A}}\llbracket 512, 160 \rrbracket$, compared to several state-of-the-art \ac{QLDPC} codes, is assessed using both \ac{BP4} and \ac{BP4-M} decoders.
For this setting, we are able to consider several code instances, one for each \ac{QLDPC} code construction\footnote{We do not consider any code based on the \ac{QC} framework of~\cite{Hagiwara2007} since, for such construction, it is not possible to construct a code with $R_{\text{Q}} \approx 0.31$.}, in the moderate block-length scenario $(n \approx 512)$, with stabilizer generator weight $\approx 8$ and $R_{\text{Q}} \approx 0.31$. 
We use the \ac{BP4-M} decoder for the \ac{dc} code families and, as expected, performance improves slightly: the gain is limited because the minimum distance of these codes is relatively small.
$\mathcal{C}_{\text{A}}\llbracket 512, 160 \rrbracket$  is characterized by better performance than $\mathcal{C}_{\text{Bic}}\llbracket 512, 160 \rrbracket$, consistently with the fact that the latter has a smaller minimum distance, namely $\widehat{d}(\mathcal{C}_{\text{Bic}}) = 2$.
Regarding $\mathcal{C}_{\text{HP}}$ and $\mathcal{C}_{\text{QT}}$, the performance under standard \ac{BP4} decoder is very similar to that of $\mathcal{C}_{\text{A}}$, even if $\mathcal{C}_{\text{HP}}$ and $\mathcal{C}_{\text{QT}}$ have a greater stabilizer generator weight, namely $10.88$ (on average) and $10$, respectively.
In fact, these three codes share the same minimum distance, i.e., $\widehat{d}(\mathcal{C}_{\text{A}}) = \widehat{d}(\mathcal{C}_{\text{HP}}) = \widehat{d}(\mathcal{C}_{\text{QT}}) = 4$.
Interestingly, using the \ac{BP4-M} decoder, for small values of $p$, the code $\mathcal{C}_{\text{A}}$ achieves slightly better performance than that of both $\mathcal{C}_{\text{HP}}$ and $\mathcal{C}_{\text{QT}}$.
The last comparison is with $\mathcal{C}_{\text{GB}}\llbracket 512, 174 \rrbracket$, which behaves better  in line with its larger girth and minimum distance. 
Under phenomenological noise, where its native \ac{PCM} provides no meta-checks, $\mathcal{C}_{\text{GB}}$ performs worse than $\mathcal{C}_{\text{A}}$ (see Section~\ref{subsec:PLnoise}).

\subsection{Performance of \emph{Construction~B} codes}

\begin{table}[tb!]
\centering
\caption{Parameters of \ac{dc} \ac{CSS} \emph{Construction~B} codes. 
}
\begin{tabular}{|l|c|c|c|c|c|c|c|c|}
\hline
\textbf{Code} & \multicolumn{7}{c|}{\textbf{Parameters}} \\
\cline{2-8}
  & $v$ & $u$ & $\ell$ & $\widehat{d}(\mathcal{C})$\hatstrut & $g$ & $n_4$ & $R_{\text{Q}}$ \\
\hline
$\mathcal{C}_{\text{B}, 3} \llbracket 128, 64 \rrbracket$ & $3$ & $4$ & $5$ & $6$ & $4$ & $192$ & $0.50$ \\
\hline
$\mathcal{C}_{\text{B}, 3} \llbracket 512, 256 \rrbracket$ & $3$ & $4$ & $7$ & $6$ & $4$ & $768$ & $0.50$ \\
\hline
$\mathcal{C}_{\text{B}, 5}^{\text{NU}} \llbracket 512, 256 \rrbracket$ & $5$ & $4$ & $7$ & $2$ & $4$ & $13\,312$ & $0.50$ \\
\hline
$\mathcal{C}_{\text{B}, 5}^{\text{C}} \llbracket 512, 256 \rrbracket$ & $5$ & $4$ & $7$ & $4$ & $4$ & $11\,520$ & $0.50$ \\
\hline
$\mathcal{C}_{\text{B}, 5}^{\text{RO1}} \llbracket 512, 256 \rrbracket$ & $5$ & $4$ & $7$ & $4$ & $4$ & $6\,400$ & $0.50$ \\
\hline
$\mathcal{C}_{\text{B}, 5}^{\text{RO2}} \llbracket 512, 256 \rrbracket$ & $5$ & $4$ & $7$ & $8$ & $4$ & $4\,096$ & $0.50$ \\
\hline
$\mathcal{C}_{\text{B}, 5}^{\text{H}} \llbracket 512, 256 \rrbracket$ & $5$ & $4$ & $7$ & $10$ & $4$ & $2\,560$ & $0.50$ \\
\hline
\end{tabular}
\label{tab:code_parameters_B_4}
\end{table}

The values of the parameters $v$, $u$, $\ell$, $\widehat{d}(\mathcal{C})$, $g$, $n_4$ and $R_{\text{Q}}$ for a number of \emph{Construction~B} codes are reported in Table~\ref{tab:code_parameters_B_4}.
The notation for such codes is $\mathcal{C}_{\text{B}, v}$.

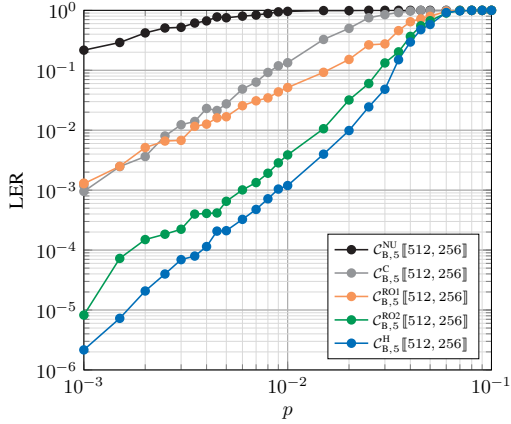
\begin{figure}[t]
    \centering
    \resizebox{0.80\columnwidth}{!}{
        \begin{tikzpicture} 
    \begin{loglogaxis}[
        xlabel={$p$},
        ylabel={\ac{LER}},
        xmin=0.001, xmax=0.1,
        ymin=0.000001, ymax=1,
        legend style={
            legend columns=1,
            font=\scriptsize
        },
        grid=both,
        major grid style={solid, gray!70},
        minor grid style={solid, gray!30},
        legend pos = south east,
        height  = 8.0cm,
        width = \columnwidth, 
        scaled x ticks=false
    ]

    \addplot[
        color = Black, 
        mark = *,        
        mark size = 2pt, 
        mark options = {solid},
        thick,
        ]
        coordinates {
            (0.00050, 1.3004e-01)
            (0.00060, 1.3459e-01)
            (0.00070, 1.5408e-01)
            (0.00080, 1.9685e-01)
            (0.00090, 2.0367e-01)
            (0.00100, 2.1645e-01)
            (0.00150, 2.9070e-01)
            (0.00200, 4.2017e-01)
            (0.00250, 5.0761e-01)
            (0.00300, 5.2083e-01)
            (0.00350, 6.1350e-01)
            (0.00400, 6.6667e-01)
            (0.00450, 7.7519e-01)
            (0.00500, 7.5188e-01)
            (0.00600, 8.0000e-01)
            (0.00700, 8.3333e-01)
            (0.00800, 8.8496e-01)
            (0.00900, 9.4340e-01)
            (0.01000, 9.6154e-01)
            (0.01500, 9.9010e-01)
            (0.02000, 9.9010e-01)
            (0.02500, 1.0000e+00)
            (0.03000, 1.0000e+00)
            (0.03500, 1.0000e+00)
            (0.04000, 1.0000e+00)
            (0.04500, 1.0000e+00)
            (0.05000, 1.0000e+00)
            (0.06000, 1.0000e+00)
            (0.07000, 1.0000e+00)
            (0.08000, 1.0000e+00)
            (0.09000, 1.0000e+00)
            (0.1, 1.0000e+00)
    };

    \addplot[
        color = Gray, 
        mark = *,        
        mark size = 2pt, 
        mark options = {solid},
        thick,
        ]
        coordinates {
            (0.00050, 2.9748e-04)
            (0.00060, 3.3393e-04)
            (0.00070, 6.2514e-04)
            (0.00080, 6.8886e-04)
            (0.00090, 8.0132e-04)
            (0.00100, 1.1867e-03)
            (0.00100, 9.5178e-04)
            (0.00150, 2.4702e-03)
            (0.00200, 3.6332e-03)
            (0.00250, 8.0090e-03)
            (0.00300, 1.2324e-02)
            (0.00350, 1.3992e-02)
            (0.00400, 2.3063e-02)
            (0.00450, 2.1317e-02)
            (0.00500, 2.7503e-02)
            (0.00600, 4.8544e-02)
            (0.00700, 6.3371e-02)
            (0.00800, 9.1996e-02)
            (0.00900, 1.1876e-01)
            (0.01000, 1.3369e-01)
            (0.01500, 3.2680e-01)
            (0.02000, 5.0000e-01)
            (0.02500, 7.5188e-01)
            (0.03000, 8.4746e-01)
            (0.03500, 9.1743e-01)
            (0.04000, 9.7087e-01)
            (0.04500, 1.0000e+00)
            (0.05000, 9.8039e-01)
            (0.06000, 1.0000e+00)
            (0.07000, 1.0000e+00)
            (0.08000, 1.0000e+00)
            (0.09000, 1.0000e+00)
            (0.10000, 1.0000e+00)
    };

    \addplot[
        color = Peach, 
        mark = *,        
        mark size = 2pt, 
        mark options = {solid},
        thick,
        ]
        coordinates {
            (0.00050, 6.2775e-04)
            (0.00060, 7.2876e-04)
            (0.00070, 9.3346e-04)
            (0.00080, 1.0215e-03)
            (0.00090, 1.1671e-03)
            (0.00100, 1.3091e-03)
            (0.00100, 1.2808e-03)
            (0.00150, 2.4933e-03)
            (0.00200, 5.1329e-03)
            (0.00250, 6.5924e-03)
            (0.00300, 6.7650e-03)
            (0.00350, 1.1590e-02)
            (0.00400, 1.2620e-02)
            (0.00450, 1.6077e-02)
            (0.00500, 1.6764e-02)
            (0.00600, 2.5621e-02)
            (0.00700, 3.0931e-02)
            (0.00800, 3.4590e-02)
            (0.00900, 4.3535e-02)
            (0.01000, 5.1520e-02)
            (0.01500, 9.2421e-02)
            (0.02000, 1.5152e-01)
            (0.02500, 2.6525e-01)
            (0.03000, 2.7548e-01)
            (0.03500, 4.5872e-01)
            (0.04000, 6.4516e-01)
            (0.04500, 7.1942e-01)
            (0.05000, 7.9365e-01)
            (0.06000, 9.6154e-01)
            (0.07000, 9.9010e-01)
            (0.08000, 1.0000e+00)
            (0.09000, 1.0000e+00)
            (0.10000, 1.0000e+00)
    };

    \addplot[
        color = ForestGreen, 
        mark = *,        
        mark size = 2pt, 
        mark options = {solid},
        thick,
        ]
        coordinates {
            (0.00050, 1.6057e-06)
            (0.00060, 1.5371e-05)
            (0.00070, 4.1692e-06)
            (0.00080, 2.6864e-05)
            (0.00090, 3.7343e-05)
            (0.00100, 8.2271e-06)
            (0.00150, 7.2615e-05)
            (0.00200, 1.4970e-04)
            (0.00250, 1.8389e-04)
            (0.00300, 2.2258e-04) 
            (0.00350, 3.9839e-04)
            (0.00400, 4.0876e-04)
            (0.00450, 4.1588e-04)
            (0.00500, 6.5128e-04)
            (0.00600, 1.0088e-03)
            (0.00700, 1.3364e-03)
            (0.00800, 1.9151e-03)
            (0.00900, 2.8497e-03)
            (0.01000, 3.8691e-03)
            (0.01500, 1.0588e-02)
            (0.02000, 3.1928e-02)
            (0.02500, 6.0132e-02)
            (0.03000, 1.3245e-01)
            (0.03500, 2.0284e-01)
            (0.04000, 3.6765e-01)
            (0.04500, 5.5249e-01)
            (0.05000, 6.7114e-01)
            (0.06000, 9.0090e-01)
            (0.07000, 9.9010e-01)
            (0.08000, 1.0000e+00)
            (0.09000, 1.0000e+00)
            (0.10000, 1.0000e+00)
    };

    \addplot[
        color = NavyBlue, 
        mark = *,        
        mark size = 2pt, 
        mark options = {solid},
        thick,
        ]
        coordinates {
            (0.00050, 1.9652e-06)
            (0.00060, 3.4850e-06)
            (0.00070, 5.2096e-06)
            (0.00080, 6.7764e-06)
            (0.00090, 3.9484e-06)
            (0.00100, 2.1606e-06)
            (0.00150, 7.2635e-06)
            (0.00200, 2.0760e-05)
            (0.00250, 3.9886e-05)
            (0.00300, 6.9238e-05)
            (0.00350, 7.9420e-05) 
            (0.00400, 1.1442e-04) 
            (0.00450, 2.0712e-04)
            (0.00500, 2.1053e-04)
            (0.00600, 3.2529e-04)
            (0.00700, 4.7650e-04)
            (0.00800, 7.1734e-04)
            (0.00900, 1.0418e-03)
            (0.010, 1.1936E-03)
            (0.015, 3.9836e-03)
            (0.020, 9.8717e-03)
            (0.025, 2.4456e-02)
            (0.030, 4.8193e-02)
            (0.035, 1.4993e-01)
            (0.040, 2.9586e-01)
            (0.045, 4.7619e-01)
            (0.050, 5.8140e-01)
            (0.060, 9.1743e-01)
            (0.070, 9.8039e-01)
            (0.080, 1.0000e+00)
            (0.090, 1.0000e+00)
            (0.100, 1.0000e+00)
        };
        

    \legend{ 
        {$\mathcal{C}_{\text{B}, 5}^{\text{NU}} \llbracket 512, 256 \rrbracket$}, 
        {$\mathcal{C}_{\text{B}, 5}^{\text{C}} \llbracket 512, 256 \rrbracket$}, 
        {$\mathcal{C}_{\text{B}, 5}^{\text{RO1}} \llbracket 512, 256 \rrbracket$}, 
        {$\mathcal{C}_{\text{B}, 5}^{\text{RO2}} \llbracket 512, 256 \rrbracket$}, 
        {$\mathcal{C}_{\text{B}, 5}^{\text{H}} \llbracket 512, 256 \rrbracket$}, 
    }
    \end{loglogaxis}
\end{tikzpicture}
    }
    \caption{\Ac{LER} of \ac{dc} \ac{CSS} \emph{Construction~B} codes, as a function of $p$, and using the \ac{BP4-M} decoder. 
    We use $u = 4$, $v = 5$, and $\ell = 7$.
    }
    \label{fig:constrB_rand_heu}
\end{figure}

Fig.~\ref{fig:constrB_rand_heu} illustrates the \ac{LER} of the $\llbracket 512, 256 \rrbracket$ codes with $v = 5$ obtained by different design choices.
For these simulations,  since we are considering only \ac{dc} \ac{CSS} codes, we employ the \ac{BP4-M} decoder.  
The following observations hold:
\begin{itemize}
    \item[$(i)$] the code $\mathcal{C}_{\text{B}, 5}^{\text{NU}}$ is characterized by the fact that its difference sets $W_i$ are not unique;
    \item[$(ii)$] the code $\mathcal{C}_{\text{B}, 5}^{\text{C}}$ has the elements of the supports of the signatures clustered in a subset of $A$; 
    \item[$(iii)$] the code $\mathcal{C}_{\text{B}, 5}^{\text{RO1}}$ is randomly generated, and it has many overlaps between its difference sets $W_i$;
    \item[$(iv)$] the code $\mathcal{C}_{\text{B}, 5}^{\text{RO2}}$ is randomly generated with a smaller number of overlaps between its difference sets $W_i$;
    \item[$(v)$] the code $\mathcal C_{\mathrm B,5}^{\mathrm H}$ is generated using the proposed heuristic procedure and achieves the lowest simulated \ac{LER} among the instances compared in this experiment.
\end{itemize}
By looking at Table~\ref{tab:code_parameters_B_4} and Fig. \ref{fig:constrB_rand_heu}, we observe that the performance of the codes strongly depends on how the $u$ signatures are selected. If two signatures have identical difference sets ($\mathcal{C}_{\text{B}, 5}^{\text{NU}}$), or if their supports occupy a small portion of $A$ ($\mathcal{C}_{\text{B}, 5}^{\text{C}}$) and exhibit many overlaps ($\mathcal{C}_{\text{B}, 5}^{\text{RO1}}$), the number of length-$4$ cycles is large. 
In fact, $\mathcal{C}_{\text{B}, 5}^{\text{NU}}$ exhibits the largest value of $n_4$ among the codes considered in this figure, namely $13\,312$ and the smallest minimum distance  $\widehat{d}(\mathcal{C}_{\text{B}, 5}^{\text{NU}}) = 2$.
Moreover, even if $\mathcal{C}_{\text{B}, 5}^{\text{C}}$ and $\mathcal{C}_{\text{B}, 5}^{\text{RO1}}$ have a lower number of length-$4$ cycles, namely $11\,520$ and $6\,400$, respectively, their minimum distance remains small, i.e., $\widehat{d}(\mathcal{C}_{\text{B}, 5}^{\text{C}}) = \widehat{d}(\mathcal{C}_{\text{B}, 5}^{\text{RO1}}) = 4$. 
In fact, the latter codes have similar, relatively poor, performance.
As the number of overlaps decreases ($\mathcal{C}_{\text{B}, 5}^{\text{RO2}}$), the number of short cycles drops ($n_4 = 4\,096$) and the minimum distance estimate increases to $\widehat{d}(\mathcal{C}_{\text{B}, 5}^{\text{RO2}}) = 8$.
As a consequence, the performance improves. 
Finally, the quantum-minimum-distance estimate of $\mathcal C_{\mathrm B,5}^{\mathrm H}$ is $\widehat d(\mathcal C_{\mathrm B,5}^{\mathrm H})=10$, matching the upper bound in Corollary~\ref{cor:min_dist_B}, and then it achieves the lowest simulated \ac{LER} in this comparison\footnote{Note that  $\mathcal{C}_{\text{B},5}^{\text{H}}\llbracket 512, 256 \rrbracket$ is the same code as the second line of Table~\ref{tab:constrB_distance}.}.

\begin{figure}[t]
    \centering
    \resizebox{0.80\columnwidth}{!}{
        \begin{tikzpicture} 
    \begin{loglogaxis}[
        xlabel={$p$},
        ylabel={ LER },
        xmin=0.0005, xmax=0.1,
        ymin=0.0000001, ymax=1,
        legend style={
            legend columns=1,
            font=\scriptsize
        },
        grid=both,
        major grid style={solid, gray!70},
        minor grid style={solid, gray!30},
        legend pos = south east,
        height  = 8.0cm,
        width = \columnwidth, 
        scaled x ticks=false
    ]
    
    \addplot[ 
        color = BrickRed, 
        mark = square, thick,
        mark options = {solid}
        ]
        coordinates {
            (0.5, 1.0)
            (0.4, 1.0)
            (0.3, 1.0)
            (0.2, 1.0)
            (0.1, 0.8849557522123894)
            (0.09, 0.8547008547008547)
            (0.08, 0.8130081300813008)
            (0.07, 0.6211180124223602)
            (0.06, 0.45662100456621)
            (0.05, 0.36363636363636365)
            (0.04, 0.1976284584980237)
            (0.03, 0.09823182711198428)
            (0.02, 0.03179650238473768)
            (0.01, 0.004429286441954201)
            (0.009, 0.002725464010247745)
            (0.008, 0.0018810428501561266)
            (0.007, 0.0010155274141625454)
            (0.006, 0.0006979438573961111)
            (0.005, 0.00031960930957996946)
            (0.004, 0.0001675614447818015)
            (0.003, 7.461294534601753e-05)
            (0.002, 2.2776718913422957e-05)
            (0.001, 2.8479975302165417e-06)
            (0.0009, 1.8992873626907828e-06)
            (0.0008, 1.4222948616653971e-06)
            (0.0007, 1.036579200735656e-06)
            (0.0006, 5.434920549342203e-07)
            (0.0005, 3.0241876371977534e-07)
        };

    \addplot[
        color = cyan, 
        mark = square, thick,
        mark options = {solid}
        ]
        coordinates { 
            (0.5, 1.0)
            (0.4, 1.0)
            (0.3, 1.0)
            (0.2, 1.0)
            (0.1, 0.8620689655172413)
            (0.09, 0.8264462809917356)
            (0.08, 0.7518796992481203)
            (0.07, 0.746268656716418)
            (0.06, 0.5780346820809249)
            (0.05, 0.4975124378109453)
            (0.04, 0.3194888178913738)
            (0.03, 0.20161290322580644)
            (0.02, 0.10460251046025104)
            (0.01, 0.029708853238265002)
            (0.009, 0.027225701061802342)
            (0.008, 0.022222222222222223)
            (0.007, 0.01893939393939394)
            (0.006, 0.014204545454545454)
            (0.005, 0.010340192327577293)
            (0.004, 0.008552125203112973)
            (0.003, 0.004570383912248629)
            (0.002, 0.002978850163836759)
            (0.001, 0.0016348270624699012)
            (0.0009, 0.00125196418541646)
            (0.0008, 0.0011010063197762754)
            (0.0007, 0.0010814785975385547)
            (0.0006, 0.0007768981564206748)
            (0.0005, 0.0006688426347049066)
        };

    \addplot[
        color = NavyBlue, 
        mark = square, thick,
        mark options = {solid}
        ]
        coordinates { 
            (0.5, 1.0)
            (0.4, 1.0)
            (0.3, 1.0)
            (0.2, 1.0)
            (0.1, 0.9174311926605505)
            (0.09, 0.9259259259259259)
            (0.08, 0.8547008547008547)
            (0.07, 0.7299270072992701)
            (0.06, 0.6756756756756757)
            (0.05, 0.5434782608695652)
            (0.04, 0.35587188612099646)
            (0.03, 0.2079002079002079)
            (0.02, 0.07457121551081283)
            (0.01, 0.021164021164021163)
            (0.009, 0.017006802721088437)
            (0.008, 0.013635124079629125)
            (0.007, 0.008476012883539583)
            (0.006, 0.007027900766041184)
            (0.005, 0.00518000518000518)
            (0.004, 0.0026360185575706454)
            (0.003, 0.0018322736683951115)
            (0.002, 0.0008440741772386957)
            (0.001, 0.00017783847997894392)
            (0.0009, 0.00014429244037475632)
            (0.0008, 9.994113467167839e-05)
            (0.0007, 8.008482584753771e-05)
            (0.0006, 7.24312265503904e-05)
            (0.0005, 4.252519724249611e-05)
        };

     \addplot[
        color = Peach, 
        mark = square, thick,
        mark options = {solid}
        ]
        coordinates {
            (0.5, 1.0)
            (0.4, 1.0)
            (0.3, 1.0)
            (0.2, 1.0)
            (0.1, 1.0)
            (0.09, 0.9615384615384616)
            (0.08, 0.9523809523809523)
            (0.07, 0.9615384615384616)
            (0.06, 0.8771929824561403)
            (0.05, 0.819672131147541)
            (0.04, 0.6666666666666666)
            (0.03, 0.6329113924050633)
            (0.02, 0.5154639175257731)
            (0.01, 0.34875621890547264)
            (0.009, 0.29411764705882354)
            (0.008, 0.24509803921568626)
            (0.007, 0.17482517482517482)
            (0.006, 0.1592356687898089)
            (0.005, 0.13368983957219252)
            (0.004, 0.12594458438287154)
            (0.003, 0.08576329331046312)
            (0.002, 0.05707762557077625)
            (0.001, 0.029489826010026542)
            (0.0009, 0.02436426116838488)
            (0.0008, 0.020907380305247754)
            (0.0007, 0.0188821752265861)
            (0.0006, 0.016877637130801686)
            (0.0005, 0.014753614635585718)
        };
        
    \addplot[
        color = CadetBlue, 
        mark = square, thick,
        mark options = {solid}
        ]
        coordinates {
            (0.5, 1.0)
            (0.4, 1.0)
            (0.3, 1.0)
            (0.2, 1.0)
            (0.1, 0.9523809523809523)
            (0.09, 0.9345794392523364)
            (0.08, 0.8849557522123894)
            (0.07, 0.7692307692307693)
            (0.06, 0.6622516556291391)
            (0.05, 0.45871559633027525)
            (0.04, 0.2824858757062147)
            (0.03, 0.1497005988023952)
            (0.02, 0.03259452411994785)
            (0.01, 0.004367575122292104)
            (0.009, 0.0025685811157916366)
            (0.008, 0.0017515588874097948)
            (0.007, 0.001175861906777668)
            (0.006, 0.000765023141950044)
            (0.005, 0.00033870060898369495)
            (0.004, 0.00019543617445414677)
            (0.003, 8.572660585221247e-05)
            (0.002, 1.9430231769633618e-05)
            (0.001, 2.3775800666730513e-06)
            (0.0009, 1.6967243650229639e-06)
            (0.0008, 1.1423568427820314e-06)
            (0.0007, 9.727972426723658e-07)
            (0.0006, 4.835329306500099e-07)
            (0.0005, 2.6869412000293727e-07)
        };

    \addplot[
        color = ForestGreen, 
        mark = square, thick,
        mark options = {solid}
        ]
        coordinates {
            (0.5, 1.0)
            (0.4, 1.0)
            (0.3, 1.0)
            (0.2, 1.0)
            (0.1, 0.9009009009009009)
            (0.09, 0.7874015748031497)
            (0.08, 0.7874015748031497)
            (0.07, 0.6024096385542169)
            (0.06, 0.4132231404958678)
            (0.05, 0.2564102564102564)
            (0.04, 0.12690355329949238)
            (0.03, 0.04111842105263158)
            (0.02, 0.010019036168720569)
            (0.01, 0.0006004491359536934)
            (0.009, 0.00035953763460190195)
            (0.008, 0.00022529118886160362)
            (0.007, 0.00015933969629853886)
            (0.006, 9.090347142176665e-05)
            (0.005, 5.7464130889498774e-05)
            (0.004, 4.1685911662550876e-05)
            (0.003, 1.8295760853914388e-05)
            (0.002, 7.385733834696641e-06)
            (0.001, 1.5397725718956768e-06)
            (0.0009, 1.1304813053550062e-06)
            (0.0008, 7.270171767695445e-07)
            (0.0007, 7.009938494519252e-07)
            (0.0006, 6.249873086303601e-07)
            (0.0005, 3.7194040244204464e-07)
        };


    \legend{ 
    {$\mathcal{C}_{\text{B}, 3} \llbracket 128, 64 \rrbracket$},
    {$\mathcal{C}_{\text{Bic}} \llbracket 128, 64 \rrbracket$},
    {$\mathcal{C}_{\text{QT}} \llbracket 168, 78 \rrbracket$},
    {$\mathcal{C}_{\text{HP}} \llbracket 130, 64 \rrbracket$},
    {$\mathcal{C}_{\text{QC}} \llbracket 156, 82 \rrbracket$},
    {$\mathcal{C}_{\text{GB}} \llbracket 168, 76 \rrbracket$},
    }
    
    \end{loglogaxis}
\end{tikzpicture}
    }
    \caption{Comparison between the \ac{LER} of $\mathcal{C}_{\text{B,}3} \llbracket 128, 64 \rrbracket$ and several other codes, as a function of $p$, and using the \ac{BP+OSD} decoder.}
    \label{fig:constrB_[[128,64,6]]}
\end{figure}
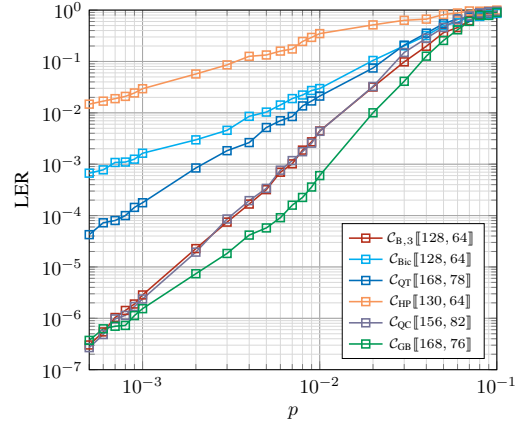

Fig.~\ref{fig:constrB_[[128,64,6]]} reports the performance of $\mathcal{C}_{\text{B}, 3}\llbracket 128, 64 \rrbracket$ and of several  state-of-the-art \ac{QLDPC} codes under \ac{BP+OSD} decoding. 
$\mathcal{C}_{\text{B}, 3}\llbracket 128, 64 \rrbracket$ is characterized by better performance than  $\mathcal{C}_{\text{Bic}}\llbracket 128, 64 \rrbracket$  and $\mathcal{C}_{\text{HP}}\llbracket 130, 64 \rrbracket$, coherent with the fact that such codes have a smaller minimum distance, namely $\widehat{d}(\mathcal{C}_{\text{Bic}}) = \widehat{d}(\mathcal{C}_{\text{HP}}) = 2$. 
The same holds also for $\mathcal{C}_{\text{QT}}\llbracket 168, 78 \rrbracket$, for which $\widehat{d}(\mathcal{C}_{\text{QT}}) = 3$. $\mathcal{C}_{\text{B}, 3}$ and $\mathcal{C}_{\text{QC}}$ show the same performance.  We attribute the same \ac{LER} to the OSD post-processing, which behaves very well in presence of length-$4$ cycles~\cite{Panteleev_degenerate}.  
The only code that shows slightly better performance than $\mathcal{C}_{\text{B}, 3}$, for $p > 10^{-3}$, is the \ac{GB} code.
Although this code has a smaller minimum distance, namely $\widehat{d}(\mathcal{C}_{\text{GB}}) = 4$, it is characterized by a smaller amount of length-$4$ cycles\footnote{Here we report only $n_{4, 1}$ since we have decoded only $\6H_1$.}, namely, $120$ instead of $192$; however, it has a smaller rate, namely $0.45$ instead of $0.50$.

\begin{figure}[t]
    \centering
    \resizebox{0.80\columnwidth}{!}{
        \input{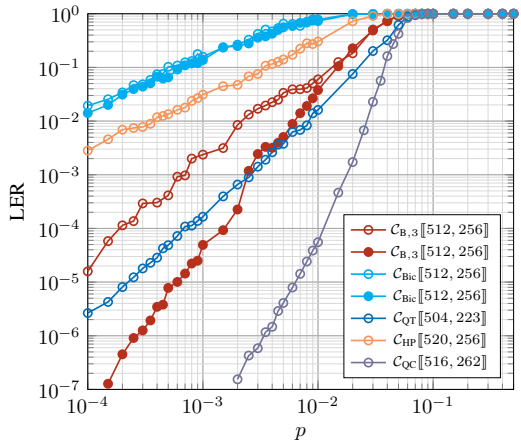}
    }
    \caption{Comparison between the \ac{LER} of $\mathcal{C}_{\text{B,}3} \llbracket 512, 256 \rrbracket$ and several other codes, as a function of $p$, using the \ac{BP4} (empty dot lines) and the \ac{BP4-M} (full dot lines) decoders.}
    \label{fig:constrB_[[512,256,6]]}
\end{figure}

Finally, Fig.~\ref{fig:constrB_[[512,256,6]]} depicts the performance of $\mathcal{C}_{\text{B}, 3}\llbracket 512, 256 \rrbracket$, compared to several instances of state-of-the-art \ac{QLDPC} codes, using both the \ac{BP4} and \ac{BP4-M} decoders.
We consider several code instances, one for each \ac{QLDPC} code construction\footnote{We do not consider any \ac{GB} code constructed as in~\cite{Panteleev_degenerate} since there are no instances with $n \approx 512$ and $R_{\text{Q}} \approx 0.50$.}, 
in the moderate block-length scenario $(n \approx 512)$, with stabilizer generator weight $\approx 12$ and $R_{\text{Q}} \approx 0.50$.
In this case, the performance of $\mathcal{C}_{\text{Bic}}$ improves only slightly.
Instead, for $\mathcal{C}_{\text{B}, 3}\llbracket 512, 256 \rrbracket$, the gain in terms of performance is much greater, in line with the larger minimum distance, in fact $\widehat{d}(\mathcal{C}_{\text{B}, 3}) = 6$. 
For instance, for $p = 3 \cdot 10^{-4}$, the \ac{LER} associated to the \ac{BP4-M} decoder is more than two orders of magnitude smaller than the \ac{LER} obtained with
\ac{BP4} decoder.
We remark that the proposed $\mathcal{C}_{\text{B}, 3}\llbracket 512, 256 \rrbracket$ exhibits better performance than  $\mathcal{C}_{\text{Bic}}\llbracket 512, 256 \rrbracket$, and $\mathcal{C}_{\text{HP}}\llbracket 520, 256 \rrbracket$. These codes have indeed a smaller minimum distance, namely $\widehat{d}(\mathcal{C}_{\text{Bic}}) = 2$ and $\widehat{d}(\mathcal{C}_{\text{HP}}) = 4$.
Regarding the comparison with code $\mathcal{C}_{\text{QT}}$, we note that the performance of the \ac{QD} code is worse with the \ac{BP4} decoder, although $\widehat{d}(\mathcal{C}_{\text{QT}}) < \widehat{d}(\mathcal{C}_{\text{B}, 3})$.
However, considering the \ac{BP4-M} decoder, the code $\mathcal{C}_{\text{B}, 3}$ behaves better even if the rate of $\mathcal{C}_{\text{QT}}$ is lower.
The code $\mathcal{C}_{\text{QC}}$, instead, exhibits more favorable performance, consistently with its higher girth and minimum distance, namely, $g = 6$, and $\widehat{d}(\mathcal{C}_{\text{QC}}) = 8$, respectively.

\subsection{Phenomenological noise}
\label{subsec:PLnoise}

In this section, we report the performance of \ac{dc} \ac{CSS} \emph{Construction~A} codes under  phenomenological noise, as a first step towards fault-tolerant quantum computing. We use the \ac{BP+OSD} decoder of Section~\ref{subsec:depo_noise}. 
The comparisons use the native \acp{PCM} of each code instance, without adding redundant stabilizer measurements. 
They therefore compare each code together with its chosen syndrome-measurement representation. 
In particular, \emph{Construction~A} provides native overcomplete \acp{PCM} that can be used for syndrome-error correction, whereas a full-row-rank representation provides no meta-check constraints. 
The observed gains under noisy readout measurements should be interpreted in this setting; they do not establish an intrinsic advantage over other codes equipped with different or augmented measurement representations.

From Theorem~\ref{theo:oneRow} and Proposition~\ref{prop:rank_deficiency}, the component \acp{PCM} of \emph{Construction~B} codes are full-rank, while those of \emph{Construction~A} codes are overcomplete. 
For this reason, and due to space limits, we restrict the numerical analysis to \emph{Construction~A} codes. 
We compute their non-null meta-check matrices as indicated in Section~\ref{app:noise}. 
Since $d_{\mathrm L}\leq U(d_{\mathrm L})\leq 5$ for all the considered codes, we compute $d_{\mathrm L}$ exactly, by exhaustive search over the syndrome-error patterns of weight up to $U(d_{\mathrm L})$. 
For both \emph{Construction~A} codes, we get $d_{\mathrm L}=4$, while the corresponding bounds $U(d_{\mathrm L})$ are $5$ and $4$, respectively.

For each \ac{CSS} code simulated in this section, in Table~\ref{tab:code_PL_noise} we collect the values of $d_{\text{L}}$ and $U(d_{\text{L}})$, together with the rank deficiency $m - r$ for the component code $C_1$, represented by $\6H_1$, and the quantum distance $\widehat{d}(\mathcal{C})$.

In Fig.~\ref{fig:constrA_PLnoise}, we show the performance of  $\mathcal{C}_{\text{A}}\llbracket 128, 24 \rrbracket$  and  $\mathcal{C}_{\text{A}}\llbracket 512, 160 \rrbracket$ under phenomenological noise and with four different readout error probabilities, namely $\epsilon \in \{ 0, \, 10^{-1}, \, 10^{-2}, \, 10^{-3} \}$.
For $\epsilon = 0$, the curves coincide with those obtained under the depolarizing-noise model. 
Moreover, the performance degrades significantly for large syndrome readout error probabilities, namely $\epsilon\in\{10^{-1}, 10^{-2}\}$, compared with the ideal-syndrome case $\epsilon=0$. 
However, for $\epsilon = 10^{-3}$, the \ac{LER} of $\mathcal{C}_{\text{A}}\llbracket 128, 24 \rrbracket$ is unaffected for $p \geq 9 \cdot 10^{-3}$.
Analogously, $\mathcal{C}_{\text{A}}\llbracket 512, 160 \rrbracket$, having the same meta-check distance $d_{\text{L}}$, performs poorly for $\epsilon \in \{ 10^{-1}, \, 10^{-2} \}$. Interestingly, for $\epsilon = 10^{-3}$, the \ac{LER} is unaffected for $p \geq 2 \cdot 10^{-4}$.

In Fig.~\ref{fig:constrA_[[128,24,8]]_PLnoise}, we compare $\mathcal{C}_{\text{A}}\llbracket 128, 24 \rrbracket$ against the same state-of-the-art codes of Fig.~\ref{fig:constrA_[[128,24,8]]}, but considering the phenomenological noise.
We fix $\epsilon = 10^{-3}$ for all simulations\footnote{We choose this value as it is comparable to what can be achieved with state-of-the-art quantum computers~\cite{Bluvstein2024}.}. 
The native \acp{PCM} used for $\mathcal C_{\mathrm{Bic}}$ and $\mathcal C_{\mathrm{HP}}$ are full-row-rank and therefore provide no meta-check constraints for identifying syndrome readout errors. Under these measurement representations, both codes exhibit a substantial performance degradation relative to the perfect-syndrome setting. 
For instance, for the code $\mathcal{C}_{\text{Bic}}$, at $p = 10^{-3}$ and $\epsilon = 0$, we get \ac{LER} of $4.9 \cdot 10^{-5}$ (see Fig.~\ref{fig:constrA_[[128,24,8]]}), while for $\epsilon = 10^{-3}$ the \ac{LER} increases to $3.0 \cdot 10^{-3}$, which corresponds to a difference of almost two orders of magnitude.
Also, the codes $\mathcal{C}_{\text{QT}}$ and $\mathcal{C}_{\text{QC}}$ both have 
$d_{\text{L}} = 2$.
As a consequence, also in this case, the \ac{LER} is significantly worse than that shown in Fig.~\ref{fig:constrA_[[128,24,8]]}.
Interestingly, we note that for low values of $p$, namely $p \leq 10^{-2}$, $\mathcal{C}_{\text{A}}$ performs slightly better than $\mathcal{C}_{\text{GB}}$, although the latter is characterized by a greater meta-check distance and  better performance for $\epsilon = 0$ (see Fig.~\ref{fig:constrA_[[128,24,8]]}).

Finally, Fig.~\ref{fig:constrA_[[512,160,4]]_PLnoise} compares
\(\mathcal{C}_{\text{A}}\llbracket 512,160 \rrbracket\) with the same state-of-the-art codes considered in Fig.~\ref{fig:constrA_[[512,160,4]]}, under phenomenological noise with \(\epsilon = 10^{-3}\), and using the \ac{BP+OSD} decoder. 
The native \acp{PCM} used for $\mathcal C_{\mathrm{Bic}}$, $\mathcal C_{\mathrm{HP}}$, and $\mathcal C_{\mathrm{GB}}$ are full-row-rank and provide no meta-check constraints. 
Their worse performance in this experiment reflects, in part, this difference in syndrome-measurement redundancy relative to $\mathcal C_{\mathrm A}$, whose native representation has $d_{\mathrm L}=4$. 
Among the benchmark codes, \(\mathcal{C}_{\text{QT}}\), exhibiting $d_{\text{L}} = 2$, is the only one with a non-null meta-check matrix. 
Notably, in this more realistic noise setting, \(\mathcal{C}_{\text{GB}}\) performs significantly worse than \(\mathcal{C}_{\text{A}}\), whereas it performs better under depolarizing noise with \ac{BP4} decoding (see Fig.~\ref{fig:constrA_[[512,160,4]]}).

\begin{table}[tb!]
\centering
\caption{Tested \ac{CSS} codes under phenomenological noise.
}
\begin{tabular}{|l|c|c|c|c|}
\hline
\textbf{Code} & \multicolumn{4}{c|}{\textbf{Parameters}} \\
\cline{2-5}
   & $m-r$ & $d_{\text{L}}$ & $U(d_{\text{L}})$ & $\widehat{d}(\mathcal{C})$\hatstrut \\
\hline                    
$\mathcal{C}_{\text{A}} \llbracket 128, 24 \rrbracket$ & $28$ & $4$ & $5$ & $8$ \\
\hline
$\mathcal{C}_{\text{A}} \llbracket 512, 160 \rrbracket$ & $80$ & $4$ & $4$ & $4$ \\
\hline
$\mathcal{C}_{\text{Bic}} \llbracket 128, 24 \rrbracket$ & $0$ & $1$ & $1$ & $4$ \\
\hline
$\mathcal{C}_{\text{Bic}} \llbracket 512, 160 \rrbracket$ & $0$ & $1$ & $1$ & $2$ \\
\hline
$\mathcal{C}_{\text{QT}} \llbracket 180, 26 \rrbracket$ & $3$ & $2$ & $2$ & $6$ \\
\hline
$\mathcal{C}_{\text{QT}} \llbracket 500, 188 \rrbracket$ & $4$ & $2$ & $2$ & $4$ \\
\hline
$\mathcal{C}_{\text{HP}} \llbracket 125, 25 \rrbracket$ & $0$ & $1$ & $1$ & $4$ \\
\hline
$\mathcal{C}_{\text{HP}} \llbracket 505, 169 \rrbracket$ & $0$ & $1$ & $3$ & $4$ \\
\hline
$\mathcal{C}_{\text{QC}} \llbracket 136, 38 \rrbracket$ & $2$ & $2$ & $3$ & $6$ \\
\hline
$\mathcal{C}_{\text{GB}} \llbracket 126, 28 \rrbracket$ & $14$ & $5$ & $5$ & $8$ \\
\hline
$\mathcal{C}_{\text{GB}} \llbracket 512, 174 \rrbracket$ & $0$ & $1$ & $1$ & $6$ \\
\hline
\end{tabular}
\label{tab:code_PL_noise}
\end{table}

\begin{figure}[t]
    \centering
    \resizebox{0.80\columnwidth}{!}{
        \input{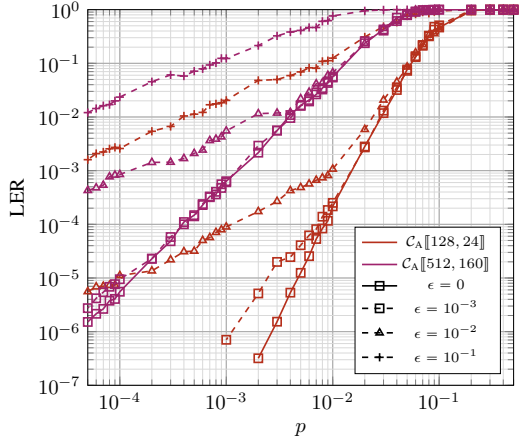}
    }
    \caption{\ac{LER} of  $\mathcal{C}_{\text{A}} \llbracket 128, 24 \rrbracket$ and of $\mathcal{C}_{\text{A}} \llbracket 512, 160 \rrbracket$, as a function of $p$, using the \ac{BP+OSD} decoder on phenomenological noise.
    We use different values of the syndrome readout error probability $\epsilon$.}
    \label{fig:constrA_PLnoise}
\end{figure}

\begin{figure}[t]
    \centering
    \resizebox{0.80\columnwidth}{!}{
        \input{Figures/constructionA_128,24,8_PLnoise}
    }
    \caption{Comparison between the \ac{LER} of $\mathcal{C}_{\text{A}} \llbracket 128, 24 \rrbracket$ and other codes, as a function of $p$, using the \ac{BP+OSD} decoder on phenomenological noise, with fixed $\epsilon = 10^{-3}$.}
    \label{fig:constrA_[[128,24,8]]_PLnoise}
\end{figure}

\begin{figure}[t]
    \centering
    \resizebox{0.80\columnwidth}{!}{
        \begin{tikzpicture} 
    \begin{loglogaxis}[
        xlabel={$p$},
        ylabel={LER},
        xmin=0.00005, xmax=0.1,
        ymin=0.000001, ymax=1.0,
        legend style={
            legend columns=1,
            font=\scriptsize
        },
        grid=both,
        major grid style={solid, gray!70},
        minor grid style={solid, gray!30},
        legend pos = south east,
        height  = 8.0cm,
        width = \columnwidth, 
        scaled x ticks=false
    ]
    
    \addplot[ 
        color = RedViolet,
        mark options = {solid},
        mark = square, thick, dashed,
        ]
        coordinates {
            (0.5, 1.0)
            (0.4, 1.0)
            (0.3, 1.0)
            (0.2, 1.0)
            (0.1, 1.0)
            (0.09, 1.0)
            (0.07, 0.9803921568627451)
            (0.08, 0.9900990099009901)
            (0.06, 0.9090909090909091)
            (0.05, 0.7874015748031497)
            (0.04, 0.6134969325153374)
            (0.03, 0.4065040650406504)
            (0.02, 0.25839793281653745)
            (0.01, 0.054318305268875614)
            (0.009, 0.04557885141294439)
            (0.008, 0.032509752925877766)
            (0.007, 0.028851702250432775)
            (0.006, 0.019477989871445268)
            (0.005, 0.01608492842206852)
            (0.004, 0.009563886763580718)
            (0.003, 0.005613562366677894)
            (0.002, 0.0029233782559125324)
            (0.001, 0.0006072382803011902)
            (0.0009, 0.00047559937411122366)
            (0.0008, 0.0004034812361051149)
            (0.0007, 0.00031990377294509813)
            (0.0006, 0.00023877061784285074)
            (0.0005, 0.0001408194849104881)
            (0.0004, 0.00010013036974140331)
            (0.0003, 5.7977635706802456e-05)
            (0.0002, 2.2964058263490064e-05)
            (0.0001, 9.484950134771656e-06)
            (9e-05, 7.67993531651279e-06)
            (8e-05, 6.855554092275621e-06)
            (7e-05, 5.444630584264941e-06)
            (6e-05, 4.180794442085315e-06)
            (5e-05, 2.741672840456958e-06)
        };

    \addplot[
        color = cyan, 
        mark = square, thick, dashed,
        mark options = {solid},
        ]
        coordinates { 
            (0.5, 1.0)
            (0.4, 1.0)
            (0.3, 1.0)
            (0.2, 1.0)
            (0.1, 1.0)
            (0.09, 1.0)
            (0.07, 1.0)
            (0.08, 0.9900990099009901)
            (0.06, 0.9615384615384616)
            (0.05, 0.9009009009009009)
            (0.04, 0.847457627118644)
            (0.03, 0.6535947712418301)
            (0.02, 0.5376344086021505)
            (0.01, 0.3257328990228013)
            (0.009, 0.28169014084507044)
            (0.008, 0.28735632183908044)
            (0.007, 0.25906735751295334)
            (0.006, 0.2457002457002457)
            (0.005, 0.21691973969631237)
            (0.004, 0.1841620626151013)
            (0.003, 0.1597444089456869)
            (0.002, 0.13679890560875513)
            (0.001, 0.08748906386701662)
            (0.0009, 0.06934812760055478)
            (0.0008, 0.06997900629811056)
            (0.0007, 0.054764512595837894)
            (0.0006, 0.05589714924538849)
            (0.0005, 0.04664179104477612)
            (0.0004, 0.03447087211306446)
            (0.0003, 0.02339728591483388)
            (0.0002, 0.01967729240456513)
            (0.0001, 0.009849305623953511)
            (9e-05, 0.008198737394441257)
            (8e-05, 0.00689797889218459)
            (7e-05, 0.00556947925368978)
            (6e-05, 0.0052027043791092916)
            (5e-05, 0.004538234626730202)
        };

    \addplot[
        color = NavyBlue, 
        mark = square, thick, dashed,
        mark options = {solid},
        mark size = 2pt, 
        ]
        coordinates { 
            (0.5, 1.0)
            (0.4, 1.0)
            (0.3, 1.0)
            (0.2, 1.0)
            (0.1, 1.0)
            (0.09, 0.9900990099009901)
            (0.07, 0.9803921568627451)
            (0.08, 1.0)
            (0.06, 0.8928571428571429)
            (0.05, 0.7751937984496124)
            (0.04, 0.5952380952380952)
            (0.03, 0.352112676056338)
            (0.02, 0.17921146953405018)
            (0.01, 0.042936882782310004)
            (0.009, 0.047058823529411764)
            (0.008, 0.037064492216456635)
            (0.007, 0.023969319271332695)
            (0.006, 0.023142791020597082)
            (0.005, 0.014945449110745778)
            (0.004, 0.014630577907827359)
            (0.003, 0.0075420469115317895)
            (0.002, 0.006101281269066504)
            (0.001, 0.0024976272541085967)
            (0.0009, 0.002043861262697488)
            (0.0008, 0.0017561421070193)
            (0.0007, 0.0015697805583202994)
            (0.0006, 0.0013503635602548797)
            (0.0005, 0.0011863662787248935)
            (0.0004, 0.0008549273739195855)
            (0.0003, 0.0006983435291488589)
            (0.0002, 0.0005973715651135006)
            (0.0001, 0.0004140872481831922)
            (9e-05, 0.00040349803985336903)
            (8e-05, 0.0003502962192022096)
            (7e-05, 0.00031877691672590607)
            (6e-05, 0.0003127912868859125)
            (5e-05, 0.0002956113538408783)
        };
        
    \addplot[
        color = Peach, 
        mark = square, thick, dashed,
        mark options = {solid},
        ]
        coordinates {
            (0.5, 1.0)
            (0.4, 1.0)
            (0.3, 1.0)
            (0.2, 1.0)
            (0.1, 1.0)
            (0.09, 1.0)
            (0.07, 0.970873786407767)
            (0.08, 1.0)
            (0.06, 0.9345794392523364)
            (0.05, 0.8928571428571429)
            (0.04, 0.7194244604316546)
            (0.03, 0.5434782608695652)
            (0.02, 0.2898550724637681)
            (0.01, 0.10111223458038422)
            (0.009, 0.09233610341643583)
            (0.008, 0.08841732979664015)
            (0.007, 0.07961783439490445)
            (0.006, 0.0648508430609598)
            (0.005, 0.045829514207149404)
            (0.004, 0.031456432840515886)
            (0.003, 0.02407318247472316)
            (0.002, 0.008696408383337682)
            (0.001, 0.0021853146853146855)
            (0.0009, 0.002169526826199206)
            (0.0008, 0.0018028412778538978)
            (0.0007, 0.0017645374403864816)
            (0.0006, 0.0013359339514254416)
            (0.0005, 0.0011792313769884788)
            (0.0004, 0.0008774084862948795)
            (0.0003, 0.0004764513900469305)
            (0.0002, 0.000289708958380411)
            (0.0001, 0.00014172777495896982)
            (9e-05, 0.00010920881489870337)
            (8e-05, 0.00010126920697096713)
            (7e-05, 7.782609913799813e-05)
            (6e-05, 7.058816885819401e-05)
            (5e-05, 6.152687555451097e-05)
        };

    \addplot[
        color = ForestGreen, 
        mark = square, thick, dashed,
        mark options = {solid},
        ]
        coordinates {
            (0.5, 1.0)
            (0.4, 1.0)
            (0.3, 1.0)
            (0.2, 1.0)
            (0.1, 0.9900990099009901)
            (0.09, 0.9615384615384616)
            (0.07, 0.7936507936507936)
            (0.08, 0.9259259259259259)
            (0.06, 0.625)
            (0.05, 0.40816326530612246)
            (0.04, 0.2127659574468085)
            (0.03, 0.09940357852882704)
            (0.02, 0.052328623757195186)
            (0.01, 0.049510075069142635)
            (0.009, 0.04177109440267335)
            (0.008, 0.04151100041511)
            (0.007, 0.033240237425804437)
            (0.006, 0.03156831151393085)
            (0.005, 0.031484035014589414)
            (0.004, 0.03135779241141424)
            (0.003, 0.024526910139356077)
            (0.002, 0.024051927326295114)
            (0.001, 0.014154281670205236)
            (0.0009, 0.013806433798149938)
            (0.0008, 0.011842728564661297)
            (0.0007, 0.011625203441060218)
            (0.0006, 0.009899029895070284)
            (0.0005, 0.008671522719389525)
            (0.0004, 0.006449948400412796)
            (0.0003, 0.005232314776056928)
            (0.0002, 0.0036023054755043226)
            (0.0001, 0.0018604124464516985)
            (9e-05, 0.0017871503887052094)
            (8e-05, 0.001753862882999807)
            (7e-05, 0.0012796560284595501)
            (6e-05, 0.0010705262707146833)
            (5e-05, 0.0008579640513062503)
        };
    

    \legend{
        {$\mathcal{C}_{\text{A}} \llbracket 512, 160 \rrbracket$},
        {$\mathcal{C}_{\text{Bic}} \llbracket 512, 160 \rrbracket$},
        {$\mathcal{C}_{\text{QT}} \llbracket 500, 188 \rrbracket$},
        {$\mathcal{C}_{\text{HP}} \llbracket 505, 169 \rrbracket$},
        {$\mathcal{C}_{\text{GB}} \llbracket 512, 174 \rrbracket$}, 
    }
    
    \end{loglogaxis}
\end{tikzpicture}
    }
    \caption{Comparison between the \ac{LER} of $\mathcal{C}_{\text{A}} \llbracket 512, 160 \rrbracket$ and other codes, as a function of $p$, using the \ac{BP+OSD} decoder on phenomenological noise, with fixed $\epsilon = 10^{-3}$.}
    \label{fig:constrA_[[512,160,4]]_PLnoise}
\end{figure}
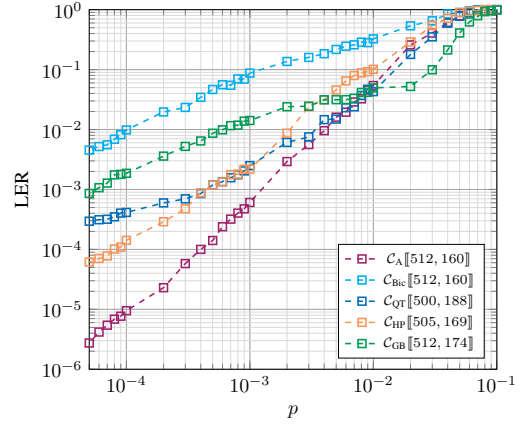

\section{Automorphism-Ensemble Decoding as a BP Rescue Procedure}
\label{sec:auto}

We use the dyadic automorphisms of the proposed codes to build an automorphism-ensemble \ac{BP} decoder, following the quantum AutDEC principle of~\cite{AutomorphismEnsembleDecoding}.
More precisely, we use the dyadic automorphism subgroup
\[
\mathcal A_{\mathrm D}
=
\left\{
\mathbf I_{n/2^\ell}\otimes\mathbf D^{(i)}
:0\leq i<2^\ell
\right\},
\]
identified in Proposition~\ref{propo:automtype}.
\ac{BP} is employed to isolate more clearly the contribution of the automorphism ensemble, since stronger constituent decoders such as \ac{BP4} or \ac{BP+OSD} could partially mask the effect of the AutDEC stage.  
For the decoding assessments in this section, we use a full-rank \ac{PCM} for the considered \ac{CSS} component codes. 
In particular, since the \ac{PCM} of \emph{Construction~A} codes is overcomplete, we first select an independent row basis and discard the remaining dependent checks. This operation preserves the rowspace, but fixes the syndrome coordinates on which the decoders are run. 
In the following, $\mathbf H$ denotes this full-rank decoding matrix, and $ \mathbf s=\mathbf H\mathbf e^T$ denotes the corresponding syndrome.  
For each dyadic automorphism $\pi$, let $\mathbf P_\pi$ denote the corresponding permutation matrix acting on the qubit coordinates. 
The automorphism also induces a linear action on the binary stabilizer checks, represented by a binary matrix $\mathbf U_\pi$ satisfying $\mathbf U_\pi\mathbf H=\mathbf H\mathbf P_\pi$.
Hence, the syndrome transformed by the automorphism is $\mathbf s_\pi = \mathbf U_\pi \mathbf s$. As detailed in Algorithm~\ref{alg:BP2_autdec_decoder}, AutDEC is activated only when the initial BP2 estimate is not syndrome-valid. Among the syndrome-valid candidates obtained using the dyadic automorphisms, the minimum-weight candidate is selected.

We consider two evaluations: the ability of AutDEC to recover BP logical failures and the LER performance of the implementable decoder in Algorithm~\ref{alg:BP2_autdec_decoder}.
For both experiments, we use $\mathcal C_{\mathrm A}\llbracket128,24\rrbracket$, analyzed in Section~\ref{subsec:performance_A} (see Table~\ref{tab:code_parameters_A}), and $\mathcal C_{\mathrm A}\llbracket128,16\rrbracket$, with $w=u=8$, $\ell=4$, $\widehat d(\mathcal C)=8$, and $R_{\mathrm Q}=0.13$. 
For both codes, the decoder uses the $|\mathcal A_{\mathrm D}|=2^\ell=16$ dyadic automorphisms.

\begin{algorithm}[t!]
\footnotesize
\KwData{\ac{PCM} $\mathbf{H}$, syndrome $\mathbf s$, depolarizing probability $p$}
\KwIn{dyadic automorphism subgroup $\mathcal A_{\mathrm D}$, maximum number of iterations $I_{\max}$}
\KwOut{estimated binary Pauli-component error vector $\widehat{\mathbf e}$}

Set $p_{\rm bin}\gets 2p/3$\;

Run BP$2$ on $(\mathbf H,\mathbf s)$ with channel parameter $p_{\rm bin}$ and obtain $\widehat{\mathbf e}_{\rm BP2}$\;

\If{$\mathbf H\widehat{\mathbf e}_{\rm BP2}^{T}=\mathbf s$}{
    \Return $\widehat{\mathbf e}_{\rm BP2}$\;
}

Set $L \gets \emptyset$\;

\For{$\pi\in\mathcal A_{\mathrm D}$}{
    Let $\mathbf P_\pi$ be the qubit-coordinate permutation induced by $\pi$\;
    Compute $\mathbf U_\pi$ such that $\mathbf U_\pi\mathbf H=\mathbf H\mathbf P_\pi$\;
    Set $\mathbf s_\pi\gets \mathbf U_\pi\mathbf s$\;

    Run BP$2$ on $(\mathbf H,\mathbf s_\pi)$ and obtain $\widehat{\mathbf e}_\pi$\;
    Map back $(\widehat{\mathbf e}^{(\pi)})^T \gets \mathbf P_\pi^{-1}\widehat{\mathbf e}_\pi^T$\;

    \If{$\mathbf H(\widehat{\mathbf e}^{(\pi)})^T=\mathbf s$}{
        $L\gets L \cup\{\widehat{\mathbf e}^{(\pi)}\}$\;
    }
}

\If{$L=\emptyset$}{
    \Return $\widehat{\mathbf e}_{\rm BP2}$\;
}
\Return $\displaystyle
\arg\min_{\mathbf e'\in L}
|\mathbf e'|$\;

\caption{BP$2$+AutDEC on depolarizing noise.}
\label{alg:BP2_autdec_decoder}
\end{algorithm}

We first evaluate the ability of AutDEC to recover BP logical failures. We run BP syndrome decoding with $I_{\max}=25$ and simulate $5\cdot10^4$ frames for each value of $p$. Using the known simulated error, we identify BP logical failures and apply the AutDEC stage to those frames, including cases in which BP returns a syndrome-valid but logically incorrect estimate. This conditioning is available only in simulations and differs from the activation rule in Algorithm~\ref{alg:BP2_autdec_decoder}. The experiment measures how often the list contains a logically correct candidate and how often the minimum-weight selection returns one, denoted by ``correct in list'' and ``selected correct'', respectively.
The results are reported in Fig.~\ref{fig:rescues64}, which shows the effectiveness of the AutDEC stage on \ac{BP} logical failures for the two codes.  
The solid curves report the conditional AutDEC success rate, namely the fraction of \ac{BP} failures for which the candidate selected by AutDEC is logically correct. 
The dashed curves report the fraction of \ac{BP} failures for which at least one  correct candidate is present in the AutDEC list.  
For the considered codes, the AutDEC list contains a correct candidate with high frequency at low and moderate depolarizing probabilities.  As $p$ increases, both conditional rates decrease, because \ac{BP} failures become more severe.

\begin{figure}[t]
    \centering
    \resizebox{0.80\columnwidth}{!}{
        \begin{tikzpicture}
    \begin{semilogxaxis}[
        xlabel={$p$},
        ylabel={Conditional AutDEC success rate},
        xmin=0.001, xmax=0.1,
        ymin=0.6, ymax=1,
        xtick={0.001,0.01,0.1},
        xticklabels={$10^{-3}$,$10^{-2}$,$10^{-1}$},
        ytick={0.6,0.7,0.8,0.9,1.0},
        grid=both,
        major grid style={solid, gray!70},
        minor grid style={solid, gray!30},
        height=8.0cm,
        width=\columnwidth,
        scaled x ticks=false,
        legend style={
            legend columns=1,
            font=\scriptsize
        },
        legend pos=south west,
        tick label style={font=\small},
        label style={font=\small}
    ]

    \addplot[
        color=NavyBlue,
        mark options={solid},
        mark=asterisk,
        mark size = 2.5pt,
        thick
    ]
    coordinates {
        (0.0025, 1.0000)
        (0.005,  0.9753)
        (0.0075, 0.9687)
        (0.01,   0.9649)
        (0.025,  0.8389)
        (0.05,   0.6144)
    };

    \addplot[
        color=NavyBlue,
        mark options={solid},
        mark=asterisk,
        mark size = 2.5pt,
        thick,
        dashed
    ]
    coordinates {
        (0.0025, 1.0000)
        (0.005,  0.9753)
        (0.0075, 0.9718)
        (0.01,   0.9691)
        (0.025,  0.8544)
        (0.05,   0.6314)
    };

  \addplot[
        color=Peach,
        mark options={solid},
        mark=asterisk,
        mark size = 2.5pt,
        thick
    ]
    coordinates {
        (0.0025, 1.0000)
        (0.005,  1)
        (0.0075, 0.981)
        (0.01,   0.9689)
        (0.025,  0.8634)
        (0.05,   0.660)
    };

    \addplot[
        color=Peach,
        mark options={solid},
        mark=asterisk,
        mark size = 2.5pt,
        thick,
        dashed
    ]
    coordinates {
        (0.0025, 1.0000)
        (0.005,  1)
        (0.0075, 0.9857)
        (0.01,   0.9728)
        (0.025,  0.8826)
        (0.05,   0.6894)
    };

   \legend{
    {$\mathcal{C}_{\mathrm A}\llbracket 128,24\rrbracket$: selected correct},
    {$\mathcal{C}_{\mathrm A}\llbracket 128,24\rrbracket$: correct in list},
    {$\mathcal{C}_{\mathrm A}\llbracket 128,16\rrbracket$: selected correct},
    {$\mathcal{C}_{\mathrm A}\llbracket 128,16\rrbracket$: correct in list}
}

    \end{semilogxaxis}
\end{tikzpicture}
    }
    \caption{Conditional AutDEC success rates for $\mathcal C_{\mathrm A}\llbracket128,24\rrbracket$ and $\mathcal C_{\mathrm A}\llbracket128,16\rrbracket$, evaluated on \ac{BP} logical failures identified using the known simulated errors. Solid curves report a correct selected candidate; dashed curves report a correct candidate anywhere in the list.}
    \label{fig:rescues64}
\end{figure}

The second experiment evaluates the \ac{LER} of Algorithm~\ref{alg:BP2_autdec_decoder}. 
Here, AutDEC is activated only when \ac{BP} fails to produce a syndrome-valid estimate. A syndrome-valid but logically incorrect estimate does not activate AutDEC and is counted as a logical error. The known simulated error is used only to assess decoding success. 
We set $I_{\max}=100$ iterations. 
For each value of $p$, frames are simulated until $100$ logical errors of the \ac{BP}$+$AutDEC decoder are observed. 
The results in Fig. \ref{fig:autdec_ler} show that the AutDEC stage is highly effective as a rescue mechanism for \ac{BP} failures, for both the tested codes. 
In particular, \ac{BP}$+$AutDEC consistently lowers the \ac{LER} with respect to standard \ac{BP} over the whole tested range of $p$. 
The gain is especially significant at low values of $p$, where most \ac{BP} failures are recovered by AutDEC. 
At larger values of $p$, the improvement is still visible, although the residual \ac{LER} increases because the list generated by the automorphism-based stage becomes less likely to contain a valid correcting candidate. 
The figure also reports the \ac{BP+OSD} performance (with the same decoder settings used in Section~\ref{sec:results}), showing that \ac{BP}$+$AutDEC is competitive with this  benchmark decoder.

\begin{figure}[tb!]
    \centering
    \centering
    \resizebox{0.80\columnwidth}{!}{
        \begin{tikzpicture}
    \begin{loglogaxis}[
        xlabel={$p$},
        ylabel={LER},
        xmin=0.002, xmax=0.2,
        ymin=0.0000001, ymax=1,
        legend style={
            legend columns=1,
            font=\scriptsize
        },
        grid=both,
        major grid style={solid, gray!70},
        minor grid style={solid, gray!30},
        legend pos=south east,
        height=8.0cm,
        width=\columnwidth,
        scaled x ticks=false
    ]


    \addplot[
        color=NavyBlue,
        mark=asterisk,
        mark size = 2.5pt,
        thick,
        solid, dotted,
        mark options={solid},
        forget plot
    ]
    coordinates {
        (0.0025, 0.000107065)
        (0.005,  0.00189627)
        (0.0075, 0.00456184)
        (0.01,   0.00394726)
        (0.025,  0.0601685)
        (0.05,   0.292398)
    };

    \addplot[
        color=NavyBlue,
        mark=asterisk,
        mark size = 2.5pt,
        thick,
        dashed, 
        forget plot,
        mark options={solid}
    ]
    coordinates {
        (0.003, 2.2E-06)
        (0.005,  4.26656e-05)
        (0.0075, 1.50633e-04)
        (0.01,   2.08184e-04)
        (0.025,  9.18864e-03)
        (0.05,   1.13895e-01)
        (0.075,  4.40529e-01)
    };

    \addplot[
        color=Peach,
        mark=asterisk,
        mark size = 2.5pt,
        thick,
        solid, dotted,
        mark options={solid},
        forget plot
    ]
    coordinates {
        (0.0025, 0.000187687)
        (0.005,  0.00125373)
        (0.0075, 0.00254346)
        (0.01,   0.00177425)
        (0.025,  0.0258732)
        (0.05,   0.163132)
    };

    \addplot[
        color=Peach,
        mark=asterisk,
        mark size = 2.5pt,
        thick,
        dashed,
        forget plot,
        mark options={solid}
    ]
    coordinates {
        (0.0025, 3.753E-07)
        (0.0075, 7.08692e-05)
        (0.0100, 0.00014194)
        (0.0250, 0.00491591)
        (0.0500, 0.0701468)
    };

    \addplot[ 
        color = NavyBlue, 
        mark options = {solid},
        mark = asterisk, thick,
        mark size = 2.5pt,
        forget plot,
        ]
        coordinates {
            (0.5,   1.0000000000000000)
            (0.4,   1.0000000000000000)
            (0.3,   1.0000000000000000)
            (0.2,   0.9900990099009901)
            (0.1,   0.5405405405405406)
            (0.09,  0.4032258064516129)
            (0.07,  0.2188183807439825)
            (0.06,  0.14992503748125938)
            (0.05,  0.08077544426494346)
            (0.04,  0.034722222222222224)
            (0.03,  0.012048192771084338)
            (0.02,  0.0026999298018251525)
            (0.01,  0.00018143617619629943)
            (0.009, 0.00011321374188398988)
            (0.008, 7.58572821094393E-05)
            (0.007, 4.4669898384915155E-05)
            (0.006, 2.361637559483746E-05)
            (0.005, 1.3639710325287996E-05)
            (0.004, 5.337832770503349E-06)
            (0.003, 1.4614357574543344E-06)
            (0.002, 3.1900960126068003E-07)
        };

    \addplot[ 
        color = Peach, 
        mark options = {solid},
        mark = asterisk, thick,
        mark size = 2.5pt,
        forget plot,
        ]
        coordinates {
            (0.5, 1.0)
            (0.4, 1.0)
            (0.3, 1.0)
            (0.2, 0.9615384615384616)
            (0.1, 0.36363636363636365)
            (0.09, 0.3076923076923077)
            (0.08, 0.17482517482517482)
            (0.07, 0.12903225806451613)
            (0.06, 0.08598452278589853)
            (0.05, 0.04336513443191674)
            (0.04, 0.02316423442205235)
            (0.03, 0.005985156811108451)
            (0.02, 0.0015309715545485166)
            (0.01, 7.563100841092444e-05)
            (0.009, 5.972902137582217e-05)
            (0.008, 4.8736696100381996e-05)
            (0.007, 2.1985439043721345e-05)
            (0.006, 1.5071008563949907e-05)
            (0.005, 8.186339422660049e-06)
            (0.004, 2.356724190140602e-06)
            (0.003, 1.0234548002240794e-06)
            (0.002, 1.3202434766086594e-07)
        };

\addlegendimage{
    draw=NavyBlue,
    fill=NavyBlue,
    thick,
}
\addlegendentry{$\mathcal{C}_{\mathrm A}\llbracket 128,24\rrbracket$}

\addlegendimage{
    draw=Peach,
    fill=Peach,
    thick,
}
\addlegendentry{$\mathcal{C}_{\mathrm A}\llbracket 128,16\rrbracket$}

\addlegendimage{
    draw=black,
    thick,
    mark=asterisk,
    mark options = {solid},
    dotted,
}
\addlegendentry{BP2}

\addlegendimage{
    draw=black,
    thick,
    mark=asterisk,
    mark options = {solid},
    dashed
}
\addlegendentry{BP2+AutDEC}

\addlegendimage{
    draw=black,
    thick,
    mark=asterisk,
    mark options = {solid},
}
\addlegendentry{BP2+OSD}

    \end{loglogaxis}
\end{tikzpicture}
    }
    \caption{Comparison between the \ac{LER} of  $\mathcal{C}_{\text A}\llbracket 128, 24\rrbracket$ and $\mathcal{C}_{\text A}\llbracket 128, 16\rrbracket$ under different decoding algorithms.}
    \label{fig:autdec_ler}
\end{figure}
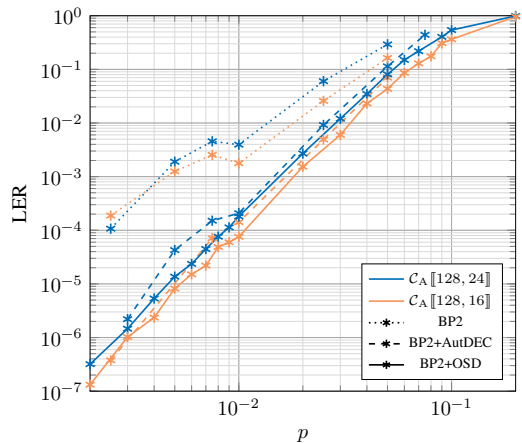

\section{Conclusions}
\label{sec:conclusions}

We have proposed two high-rate DC CSS QLDPC constructions based on QD matrices: \emph{Construction A}, which offers flexibility in code length, rate, and stabilizer generator weight, and \emph{Construction B}, which has a simpler structure and supports higher quantum rates, with quantum minimum distance bounded by twice the signature weight. For Construction B, the proposed heuristic reduces the number of length-4 cycles and, among the tested designs, yields the largest distance estimate and lowest simulated LER. 
We have also established theoretical results on the cycle properties of QD codes.

The results of Monte Carlo simulations show that the proposed \ac{QD} codes achieve comparable \ac{LER} performance to several state-of-the-art \ac{QLDPC} code families, including bicycle, \ac{QC}, \ac{HP}, \ac{GB}, and \ac{QT} codes, in the considered finite-length regimes under both depolarizing and phenomenological-noise models, the latter with the native syndrome-measurement representation of each code.
In addition, we have shown that the automorphisms induced by the \ac{QD} structure can be exploited at the decoder side through a \ac{BP}$+$AutDEC procedure in which each dyadic automorphism is applied consistently to both qubit and stabilizer-check coordinates through the induced check action.

The implementation-related features discussed here are the preservation of the code space under transversal physical Hadamard gates, the availability of the \ac{BP4-M} decoder for the \ac{dc} \ac{CSS} check structure, and the explicit dyadic permutations used in AutDEC. 
Their practical value must be assessed together with stabilizer generator weight and syndrome-extraction cost.
Future work will therefore address syndrome-extraction schedules and circuit-level performance under specified connectivity constraints, as well as the induced action of the dyadic automorphisms on the logical operators.

\bibliographystyle{IEEEtran}
\bibliography{strings, References_short}

\end{document}